\documentclass[11pt]{article}

\usepackage{fullpage,amsmath,amsthm,amssymb,thm-restate,times}
\newtheorem{theorem}{Theorem}[section]
\newtheorem{lemma}[theorem]{Lemma}
\newtheorem{claim}[theorem]{Claim}

\newtheorem{definition}[theorem]{Definition}

\usepackage{float}

\usepackage{todonotes}
\renewcommand{\todo}[1]{}

\usepackage[utf8]{inputenc}

\usepackage{microtype,hyperref,graphicx}
   \hypersetup{%
      breaklinks,%
      colorlinks=true,%
      urlcolor=[rgb]{0.25,0.0,0.0},%
      linkcolor=[rgb]{0.5,0.0,0.0},%
      citecolor=[rgb]{0,0.2,0.445},%
      filecolor=[rgb]{0,0,0.4},
      anchorcolor=[rgb]={0.0,0.1,0.2}%
   }
   
\title{A Linear Time Algorithm for the Maximum Overlap of Two Convex Polygons Under Translation}

\author{Timothy M. Chan\thanks{
Siebel School of Computing and Data Science, University of Illinois at Urbana-Champaign, USA (tmc@illinois.edu).
Work supported in part by NSF Grant CCF-2224271.
} \and
Isaac M. Hair\thanks{
Department of Computer Science, University of California, Santa Barbara, USA (hair@ucsb.edu).
}
}

\newtheorem{problem}{Problem}

\newcommand{\R}{\mathbb{R}}

\newcommand{\eps}{\varepsilon}

\begin{document}

\maketitle

\begin{abstract}
    Given two convex polygons $P$ and $Q$ with $n$ and $m$ edges, the \emph{maximum overlap} problem is to find a translation of $P$ that maximizes the area of its intersection with $Q$. We give the first randomized algorithm for this problem with linear running time. Our result improves the previous two-and-a-half-decades-old algorithm by de Berg, Cheong, Devillers, van Kreveld, and Teillaud (1998), which ran in $O((n+m)\log(n+m))$ time, as well as multiple recent algorithms given for special cases of the problem.
\end{abstract}

\section{Introduction}
%\textcolor{red}{Need a very compelling story for the introduction... also motivate as a matching problem but don't actually call it this.} 

Problems related to convex polygons are widely studied in computational geometry, as they are fundamental and give insights to more complex problems. We consider one of the most basic problems in this class:
%and useful questions:

\begin{problem} \label{prob:overlap}
    Given two convex polygons $P$ and $Q$ in the plane, with $n$ and $m$ edges respectively, find a vector $t \in \mathbb{R}^2$ that maximizes the area of $(P + t) \cap Q$, where $P + t$ denotes $P$ translated by $t$.
\end{problem}

This problem for convex polygons was first explicitly posed as an open question by Mount, Silverman, and Wu~\cite{MountSW96}.
In 1998, de Berg, Cheong, Devillers, van Kreveld, and Teillaud~\cite{BergCDKT98} %ISAAC96
presented the first efficient algorithm which solves the problem in $O((n+m)\log(n+m))$ time.
As many problems have $n\log n$ complexity, this would appear to be the end of the story
for the above problem---or is it?

%showed that we can extend beyond the naive application of parametric search to get an algorithm running in time $O((n+m)\log(n+m))$ 

\paragraph{Motivation and Background.}
The computational geometry literature is replete with various types of problems about simple and convex polygons~\cite{ORourkeST17},
e.g., simple-polygon containment~\cite{Chazelle83} (which has seen exciting development even in the convex polygon case, with translations
and rotations, as recently as last year's SoCG~\cite{ChanH24}). The maximum overlap problem may be viewed as an outgrowth of 
convex-polygon containment problems.

The maximum overlap problem may also be viewed as a formulation of \emph{shape matching}~\cite{AltG00}, which has numerous
applications and has long been a popular topic in computational geometry.  Many notions of distance between
shapes have been used in the past (e.g., Hausdorff distance, Fr\'echet distance, etc.), and the overlap area is
another very natural measure (the larger it is, the closer the two shapes are) and has been considered in numerous papers for different classes of objects (e.g., two convex polytopes in higher dimensions~\cite{AhnBS08,AhnCR13,AhnCKY14}, two arbitrary simple polygons in the plane~\cite{ChengL13,Har-PeledR17}, two unions of balls \cite{CabelloBGKOV09}, one convex
polygon vs.\ a discrete point set\footnote{In the discrete case, we are maximizing the number of points inside the translated convex polygon.} \cite{BarequetDP97,AgarwalHRSSW02}, etc.)\ as well as different types of
motions allowed (translations, rotations, and scaling---with translations-only being the most often studied).  Problem~\ref{prob:overlap} is perhaps the simplest and most basic version, but is already quite fascinating from the theoretical or technical perspective, as we will see.  

\paragraph{A Superlinear Barrier?}

Due to the concavity of (the square root of) the objective function~\cite{BergCDKT98},
it is not difficult to obtain an $O((n+m)\mathop{\rm polylog}(n+m))$-time algorithm
for the problem by applying the well-known \emph{parametric search} technique~\cite{Megiddo83,AgarwalS98}
(or more precisely, a multidimensional version of parametric search~\cite{Toledo92}).
De Berg, Cheong, Devillers, van Kreveld, and Teillaud~\cite{BergCDKT98} took more care in eliminating extra logarithmic factors to obtain their $O((n+m)\log(n+m))$-time algorithm, by avoiding parametric search and instead using more elementary binary search and matrix search techniques~\cite{FredericksonJ84}.
But any form of binary search would seem to generate at least one logarithmic factor, and so it is tempting to believe that
their time bound might be the best possible.

On the other hand, in the geometric optimization literature, there is another well-known technique that can
give rise to \emph{linear}-time algorithms for various problems: namely,
\emph{prune-and-search} (pioneered by Megiddo and Dyer~\cite{Megiddo83a,Dyer84,Megiddo84}).  For prune-and-search to be applicable, one needs a way to
prune a fraction of the input elements at each iteration, but our problem is sufficiently complex that
it is unclear how one could throw away any vertex from the input convex polygons and preserve the answer.

Perhaps because of these reasons, no improved algorithms have been reported since de Berg et al.'s 1998 work.
This is not because of lack of interest in finding faster algorithms.
For example, a (lesser known) paper by Kim, Choi, and Ahn~\cite{KimCA21} considered an unbalanced case of Problem~\ref{prob:overlap} and gave an algorithm with $O(n+m^2\log^3n)$ running time, which is linear
when $m \leq O\left(\sqrt{n/\log^3n}\right)$.  Ahn et al.~\cite{AhnCPSV07} investigated \emph{approximation} algorithms for the problem and obtained logarithmic
running time for approximation factor arbitrarily close to~1, while Har-Peled and Roy~\cite{Har-PeledR17} more generally gave a linear-time algorithm to approximate the entire objective function, which was then used to obtain a linear-time approximation algorithm for an extension of the problem to simple polygons that are decomposable into $O(1)$ convex pieces.

\paragraph{Our Result.} 
We obtain a randomized (exact) algorithm for Problem \ref{prob:overlap}, running in \emph{linear} (i.e., $O(n+m)$) expected time! 
This is the first improvement to de Berg et al.'s result~\cite{BergCDKT98} in 26 years, and is clearly optimal, and so fully settles the complexity of Problem \ref{prob:overlap}.

\paragraph{A New Kind of Prune-and-Search?}
Besides our result itself, the way we obtain linear time is also interesting.
As mentioned, if one does not care about extra $\log (n+m)$ factors, one can use
(multidimensional) parametric search to solve the problem.  Specifically, parametric search allows
us to reduce the optimization problem to the implementation of a \emph{line oracle}
(determining which side of a given line the optimal point $t^*$ lies in), which in turn
can be reduced to a \emph{point oracle}
(evaluating the area of $(P+t)\cap Q$ for a given point $t$).
A point oracle can be implemented in linear time by computing the intersection $(P+t)\cap Q$ explicitly.
(In de Berg et al.'s work~\cite{BergCDKT98}, they directly implemented the line oracle in linear time.)

Various linear-time prune-and-search algorithms (e.g., for low-dimensional linear programming~\cite{Megiddo83a,Dyer84,Megiddo84}, ham-sandwich cuts~\cite{Megiddo85},
centerpoints~\cite{JadhavM94}, Euclidean 1-center~\cite{Dyer86}, rectilinear 1-median~\cite{Zemel84,OgryczakT03}, etc.)\ also exploit line or point oracles; using (in modern parlance)
\emph{cuttings}~\cite{Chazelle04}, such oracle calls let them refine the search space and eliminate
a fraction of the input at each iteration.
Unfortunately, for our problem, we can't technically eliminate any part of the input convex polygons $P$ and $Q$.
Instead, our new idea is to reduce the \emph{cost} of the oracles iteratively as the search space is refined, 
which effectively is as good as reducing the input complexity itself.
As the algorithm progresses, oracles get cheaper and cheaper, and we will
bound the total running time by a geometric %super-geometric 
series.

To execute this strategy, one should think of oracles as data structures with sublinear query time (i.e.,
sublinear in the original input size $n+m$).  Our key idea to designing such data structures is to divide
the input convex polygons into \emph{blocks} of a certain size based on angles/slopes.  This idea is inspired by the recent
work of Chan and Hair~\cite{ChanH24} (they studied the convex-polygon containment problem 
with both translations and rotations, and although our problem without rotations might seem different,
their divide-and-conquer algorithms also crucially relied on division of the input convex polygons by angles/slopes).
We show that at each iteration, as we refine the search space, we can use the current ``block data structure'' to
obtain a better ``block data structure'' with a larger block size.
(The refinement process from one block structure to another block structure shares
some general similarities with the idea of \emph{fractional cascading}~\cite{ChazelleG86}.)

We are not aware of any prior algorithms in the computational geometry literature that work in %exactly
the same way. %(which makes our techniques interesting).
A linear-time algorithm by Chan~\cite{Chan21} for a matrix searching problem
(reporting all elements in a Monge matrix less than given value)
has a recurrence similar to the one we obtain here, but this appears to be a coincidence.  An algorithm for computing so-called ``$\eps$-approximations'' in range spaces by Matou\v sek~\cite{Matousek95} also achieved linear time (for constant $\eps$) by iteratively increasing block sizes, but the algorithm still operates in the realm of traditional prune-and-search (reducing actual input size).
Chazelle's linear-time algorithm for triangulating simple polygons~\cite{Chazelle91} also iteratively
refines a ``granularity'' parameter analogous to our block size, and as one of many steps, also requires the implementation of oracles (for ray shooting, in his case) whose cost similarly varies as a function of the granularity
(at least from a superficial reading of his paper)---fortunately, our algorithm here will not be as complicated!

\paragraph{Remarks.}
Recently, Jung, Kang, and Ahn~\cite{JungKA25}
have announced a \emph{linear}-time algorithm for a related problem:
given two convex polygons $P$ and $Q$ in the plane,
find a vector $t\in\R^2$ that minimizes the area of the
convex hull $(P+t)\cup Q$.
Although the problem may look similar to our problem on the surface, it is in many ways different:
in the minimum
convex hull problem, the optimal point $t^*$ is known to lie on one of $O(n+m)$ candidate lines 
(formed by an edge of one convex polygon and a corresponding extreme vertex of the
other convex polygon), whereas in the maximum overlap problem
(Problem~\ref{prob:overlap}), $t^*$
lies on one of $O((n+m)^2)$ candidate lines (defined by pairs of edges of the
two convex polygons).
%This property enables
Thus a more traditional prune-and-search approach can be used to solve
the minimum convex hull problem, but not the maximum overlap
problem.
Indeed, Jung et al.'s paper~\cite{JungKA25} considered the maximum overlap problem,
but they do not give an improvement over de Berg et al.'s $O((n+m)\log(n+m))$ time algorithm for the two-dimensional case. %it appears they were not able to get any new result in the two-dimensional case---they presented
Instead, they give some new results for the problem in higher dimensions $d$, but their time bound is %there is
$O(n^{2\lfloor d/2\rfloor})$, which is much larger than linear.

We observe that, in the specific case of maximum overlap for convex polytopes in dimension $d = 3$, %should also remark that in the generalization to $d=3$ dimensions (maximizing the volume of overlap 
%for two convex polytopes in $\R^3$ under translations), it is actually straightforward to
one can obtain
a significantly faster $O((n+m)\mathop{\rm polylog} (n+m))$-time algorithm by directly using 
multi-dimensional parametric search~\cite{Toledo92}, since a point oracle
(computing the volume of the intersection of two convex polytopes in $\R^3$) can
be done in $O(n+m)$ time by explicitly computing the intersection~\cite{Chazelle91,Chan16}, and
this is parallelizable with $O((n+m)\mathop{\rm polylog} (n+m))$ work and $O(\mathop{\rm polylog}(n+m))$ steps by known parallel algorithms for intersection of halfspaces 
in $\R^3$, i.e., convex hull in $\R^3$ in the dual (e.g., see \cite{AmatoP95}). This observation seems to have not appeared in the literature before (including \cite{ZhuK23, JungKA25}).

\section{Preliminaries}
We assume that the input convex polygons $P$ and $Q$ have no (adjacent) parallel edges, as these can be merged. Without loss of generality, $P$ and $Q$ both contain the origin as an interior point. Throughout, we assume that the edges of $P$ and $Q$ are given in counterclockwise order, and we let $N := n + m$. Denote by $[x]$ the set $\{1, 2, \ldots, x\}$, and denote by $[x, y]$ the set $\{x, x+1, \ldots, y\}$. % when $x \leq y$, with $[x,y]=\emptyset$ when $x>y$. %If $x$ is not an integer, then we use $[x]$ to denote the set $\{1, \ldots, \lceil x\rceil\}$, i.e., $x$ is implicitly rounded up to the nearest integer.
``\emph{With high probability}'' means ``with probability at least $1 - N^{-\omega(1)}$.'' 
For a polygonal chain $X$, let $\vert X\vert$ denote the number of vertices in the chain; for an interval $\mathcal{I} \subset [n]$, let $\vert \mathcal{I}\vert$ denote the number of integers in the interval. %Empty intervals and empty chains are handled in the natural way: the corresponding union of blocks has zero area, and angle-range conditions involving an empty chain are vacuous.
In running-time bounds, $\log x$ should be read as $\log\max\{2, x\}$. We use $\uplus$ to denote the additive union; for example, $\{1, 2\} \uplus \{2, 3\} = \{1, 2, 2, 3\}$.

The \emph{interior} of a triangle refers to all points strictly inside of its bounding edges, and the \emph{interior} of a line segment refers to all points strictly between its bounding points. %For a triangle or line segment $\mathcal{T} \subset \mathbb{R}^2$, let $\textnormal{int}(\mathcal{T})$ denote its interior.
We say that a line $\mathcal{L}$ \emph{strictly intersects} a triangle $\mathcal{T}$ if $\mathcal{L}$ intersects the interior of $\mathcal{T}$. We say that a line $\mathcal{L}$ \emph{strictly intersects} a line segment $\mathcal{T}$ if $\mathcal{L}$ intersects the interior of $\mathcal{T}$, \emph{and also} $\mathcal{L}$ does not contain $\mathcal{T}$.

\paragraph{The Objective Function.}
We want to find a translation $t \in \mathbb{R}^2$ maximizing \[\textnormal{Area}((P+t)\cap Q).\] De Berg, Cheong, Devillers, van Kreveld, and Teillaud showed several interesting properties of the area as a function of $t$ \cite{BergCDKT98}. For our purposes, we only need to import the following.\footnote{Lemma \ref{lem:objconcave} is proved by constructing a 3-dimensional convex polytope modeling the intersection of $P$ and $Q$ as $P$ is translated along a line. The polytope is then analyzed using a theorem due to Gr{\"u}nbaum, Klee, Perles, and Shephard \cite{grunbaum1967convex}.}

\begin{lemma} \label{lem:objconcave}
    $\sqrt{\textnormal{Area}((P+t)\cap Q)}$ is concave over %the convex domain of
    all values of $t$ such that $\textnormal{Area}((P+t)\cap Q) > 0$.
\end{lemma}

The restriction on $t$ is important, as $\textnormal{Area}((P+t)\cap Q)$ suddenly transitions from a concave function to a constant function $\textnormal{Area}((P+t)\cap Q) = 0$ when $t$ is outside of the specified region. For our algorithm, however, this technicality will have little impact. %; all we need is that Lemma \ref{lem:objconcave} implies $\textnormal{Area}((P+t)\cap Q)$ is unimodal.

\paragraph{Cuttings.}
Our algorithm will leverage low-dimensional cuttings.

\begin{definition}[$\varepsilon$-Cuttings]
    An \emph{$\varepsilon$-cutting} for a (multi)set $X$ of $n$ hyperplanes in $\mathbb{R}^d$, for any $d \geq 1$, is a partition of $\mathbb{R}^d$ into simplices such that the interior of every simplex intersects at most $\varepsilon \cdot n$ hyperplanes of $X$.
\end{definition}

Endowed with the ability to sample random hyperplanes, we can actually produce such a cutting with high probability in sublinear time, as per Clarkson and Shor~\cite{Clarkson87,ClarksonS89}. In the case of constant $\varepsilon$, this follows roughly by sampling logarithmically many random hyperplanes and returning an $(\varepsilon/2)$-cutting for the sample.

\begin{lemma} \label{cor:approxcut}
    Let $X$ be a (multi)set of $n$ hyperplanes in $\mathbb{R}^d$, for any constant $d \geq 1$. Assume we can sample a uniformly random member of $X$ in expected time $O(n^{0.1})$. Let $\varepsilon$ be any constant. Then in $O(n^{0.1+o(1)})$ expected time, we can find a constant-size simplicial subdivision (allowing unbounded cells) that constitutes an $\varepsilon$-cutting of $X$ with high probability.
\end{lemma}

\paragraph{Area Prefix Sums.}
Let $v_1, \ldots, v_n$ be the vertices of $P$ listed in counterclockwise order, and let $v_1', \ldots, v_m'$ be the vertices of $Q$ listed in counterclockwise order.\footnote{$v_1$ and $v_1'$ can be chosen arbitrarily, so long as they are fixed throughout the algorithm.} Let $P[v_x:v_y]$ denote the simple polygon enclosed by (i) the line segment from the origin to $v_x$, (ii) the line segment from the origin to $v_y$, and (iii) the edges encountered when traversing the boundary of $P$ in counterclockwise order from $v_x$ to $v_y$. (Define $P[v_x:v_x] := \emptyset$ for all $x \in [n]$.) We define $Q[v_x':v_y']$ analogously.

By a standard application of prefix sums, we can produce a data structure to report $\textnormal{Area}(P[v_x:v_y])$ and $\textnormal{Area}(Q[v_x':v_y'])$ with $O(N)$ preprocessing time and $O(1)$ query time. This is because (taking $\textnormal{Area}(P[v_x:v_y])$ as an example):
\begin{align*}
    \textnormal{Area}(P[v_x:v_y]) =\ & \textnormal{Area}(P[v_1:v_y]) - \textnormal{Area}(P[v_1:v_x]) \textnormal{ when $y \geq x$, and} \\
    \textnormal{Area}(P[v_x:v_y]) =\ & \textnormal{Area}(P) - (\textnormal{Area}(P[v_1:v_x]) - \textnormal{Area}(P[v_1:v_y])) \textnormal{ when $y < x$.}
\end{align*}
Throughout, we refer to the above data structure as the \emph{area prefix sums for $P$ and $Q$}.

\paragraph{The Configuration Space.}
Consider two simple polygons $X$ and $Y$, or more generally two polygonal chains $X$ and $Y$. Let $e$ be an edge of $X$ and let $e'$ be an edge of $Y$. These two edges determine a parallelogram $\pi(e, e')$ in configuration space such that, for any vector $t \in \mathbb{R}^2$, we have that $e+t$ intersects $e'$ iff $t \in \pi(e, e')$. Let $\Pi(e, e')$ denote the set of four line segments enclosing $\pi(e, e')$. Define the set of all such line segments as
\begin{equation} \label{eq:Pi}
    \Pi(X, Y)\ := \bigcup_{\textnormal{edge $e$ of $X$, edge $e'$ of $Y$}} \Pi(e, e').
\end{equation}
Define $\overleftrightarrow{\Pi}(X, Y)$ to be $\Pi(X, Y)$ where each line segment is extended to a line. Now let $\mathcal{T} \subset \mathbb{R}^2$ be any triangle or line segment. We say that $(X+t) \cap Y$ has the same \emph{combinatorial description} for all $t$ in the interior of $\mathcal{T}$ if and only if, for all edges $e$ of $X$ and $e'$ of $Y$, one of the following holds:
    \begin{enumerate}
        \item $\mathcal{T} \subseteq \pi(e, e')$.
        \item $\pi(e, e')$ is disjoint from the interior of $\mathcal{T}$.%\footnote{Recall that the interior of a triangle is the set of all points strictly inside its bounding edges, and the interior of a line segment is the set of all points strictly between its bounding points.}
    \end{enumerate}
Intuitively, this means that for all translations $t$ in the interior of $\mathcal{T}$, $(X+t)\cap Y$ will have the same pairs of edges intersecting.

Our algorithm will use a result of Mount, Silverman, and Wu \cite{MountSW96}, stated below in a simplified form:

\begin{lemma} \label{lem:existsquadratic}
    Let $A$ and $B$ be any simple polygons, and let $\mathcal{T} \subset \mathbb{R}^2$ be any triangle or line segment. Suppose that $(A+t) \cap B$ has the same combinatorial description for all $t$ in the interior of $\mathcal{T}$. Then there is a %single
    quadratic function $f(t)$ such that $f(t) = \textnormal{Area}((A+t) \cap B)$ for all $t \in \mathcal{T}$.\footnote{The quadratic function also applies to points on the boundary of $\mathcal{T}$, by continuity. 
    } Suppose that we are additionally given the list \[L = \{(e, e') : e \textnormal{ is an edge of } A, e' \textnormal{ is an edge of } B, \mathcal{T} \subseteq \pi(e, e')\},\]
    and we have access in $O(1)$ query time to the endpoints of each edge of $A$ and $B$, oriented so that the relevant interior lies to the left. Then the gradient field of $f(t)$ can be computed in time $O(\vert L\vert)$. 
\end{lemma}

\paragraph{Angles and Angle Ranges.}
The \emph{angle} of an edge $e$ on the boundary of an arbitrary simple polygon $A$, denoted $\theta_A(e)$, is the angle (measured counterclockwise) starting from a horizontal rightward ray and ending with a ray coincident with $e$ that has the interior of $A$ on its left side. We always have that $\theta_A(e) \in [0, 2\pi)$. For any subset $X$ of the edges of $A$, let $\Lambda_A(X)$ denote the minimal interval in $[0, 2\pi)$ such that $\theta_A(e) \in \Lambda_A(X)$ for all $e \in X$. We call $\Lambda_A(X)$ the \emph{angle range} for the subset $X$ with respect to $A$.

%\paragraph{Angles and Angle Ranges.}
%The \emph{angle} of an edge $e$ on the boundary of an arbitrary simple polygon $A$, denoted $\theta_A(e)$, is the angle (measured counterclockwise) starting from a horizontal rightward ray and ending with a ray coincident with $e$ that has the interior of $A$ on its left side. We always have that $\theta_A(e) \in [0, 2\pi)$. For any %connected
%subset $X$ of the edges of $A$, let $\Lambda_A(X)$ denote the corresponding circular interval of angles in $[0,2\pi)$, choosing the representation that follows the counterclockwise order of the chain. If this interval crosses angle $0$, we split the chain at that angle and handle the resulting $O(1)$ chains separately; throughout the paper, statements using $\Lambda_A(X) \cap \Lambda_B(Y)=\emptyset$ are applied after this constant-size split. We call $\Lambda_A(X)$ the \emph{angle range} for the subset $X$ with respect to $A$.

%\footnote{Recall that the edges of $P$ and $Q$ are given in counterclockwise order. This means that, at the onset of our algorithm, the set can be sorted once and for all in linear time.}\footnote{We assume that no angles are shared between the edges of $P$ and $Q$. This case can be handled easily but requires a clumsier definition.}

\section{Blocking Scheme} \label{sec:blockstructure}
Define the multiset
\[U = \{\theta_P(e) : e \textnormal{ is an edge of $P$}\} \,\cup\, \{\theta_Q(e) : e \textnormal{ is an edge of $Q$}\},\]
and sort $U$ in increasing order. Each member of $U$ has multiplicity at most $2$. This is because there might exist an edge $e$ of $P$ and an edge $e'$ of $Q$ such that $\theta_P(e) = \theta_Q(e')$, but $P$ and $Q$ themselves don't have (adjacent) parallel edges. Note that $U$ can be sorted in linear time. This is because $P$ and $Q$ are both convex polygons, and we assumed that their edges are given in counterclockwise order, so we already have access to $\{\theta_P(e) : e \textnormal{ is an edge of $P$}\}$ and $\{\theta_Q(e) : e \textnormal{ is an edge of $Q$}\}$ in sorted order.

Given a parameter $b \in \mathbb{N}$, referred to as the \emph{block size}, we partition $P$ and $Q$ into \emph{blocks} $P_1, \ldots, P_{\lceil N/b\rceil}$ and $Q_1, \ldots, Q_{\lceil N/b\rceil}$ as follows.\footnote{It is not quite a partition because the boundaries overlap, but these have zero area.} First partition $U$ into subsequences $U_1, \ldots, U_{\lceil N/b\rceil}$ of length $b$ each, except possibly the last one, such that $U_1$ contains the smallest $b$ elements, $U_2$ contains the next smallest $b$ elements, etc. Now for each $i \in \{2, \ldots, \lceil N/b\rceil\}$, check whether $U_i \cap U_{i-1} \neq \emptyset$. This can only happen when the smallest element of $U_i$ also appears in $U_{i-1}$. For every such index $i$, delete the smallest element of $U_i$. This ensures that for every edge $e$ of $P$, there is exactly one index $i$ such that $\theta_P(e) \in U_i$, and similarly for $Q$.

For each $i \in [\lceil N/b\rceil ]$, let $Z^{(P)}_i$ be the subset of $P$'s edges such that $e \in Z_i^{(P)}$ if and only if $\theta_P(e) \in U_i$. Because $P$ is a convex polygon, $Z^{(P)}_i$ is a connected polygonal chain whenever it is nonempty. Now for each $i \in [\lceil N/b\rceil]$, define the block $P_i$ as follows. If $Z_i^{(P)}$ is empty, then $P_i$ is the single point at the origin; otherwise, $P_i$ is the simple polygon enclosed by: (i) the line segment from the origin to the very first vertex of $Z_i^{(P)}$, (ii) the line segment from the origin to the very last vertex of $Z_i^{(P)}$, and (iii) the polygonal chain $Z_i^{(P)}$ itself. The blocks $Q_j$ are defined similarly. See Figure \ref{fig:blocks} for an example. %Degenerate blocks have zero area and may be ignored by the intersection routines except for their boundary endpoints.

%Also notice that the edges of a block $P_i$ (along with their orientation with respect to the interior of the block) can be found on-demand, without needing to construct the block partition explicitly. The same observation holds for each block $Q_j$.

%We then define $P_i$ as the convex hull of the origin and the set of edges\footnote{The set of edges may be empty, in which case the block is just a single point (the origin).} $\{e : e \textnormal{ is an edge of $P$,}$ and $\theta_P(e) \in S_i\}$. 

%It turns out that 
Notice that the blocks indeed partition $P$ and $Q$, because we assume that the origin is contained within both convex polygons. 
Since $b$ uniquely determines the partition, we refer to $P_1, \ldots, P_{\lceil N/b\rceil}$ and $Q_1, \ldots, Q_{\lceil N/b\rceil}$ as the \emph{blocks determined by parameter} $b$. The blocking scheme is useful because the angle ranges $\Lambda_P(E_P(P_i))$ and $\Lambda_Q(E_Q(Q_j))$ are strictly disjoint whenever $i \neq j$, where $E_P(P_i)$ is the set of edges that $P_i$ shares with $P$, and $E_Q(Q_j)$ is the set of edges that $Q_j$ shares with $Q$. We call $E_P(P_i)$ the \emph{restriction of $P_i$'s edges to $P$}, and similarly for $E_Q(Q_j)$.

\begin{figure}
    \centering
    \includegraphics[width=\textwidth]{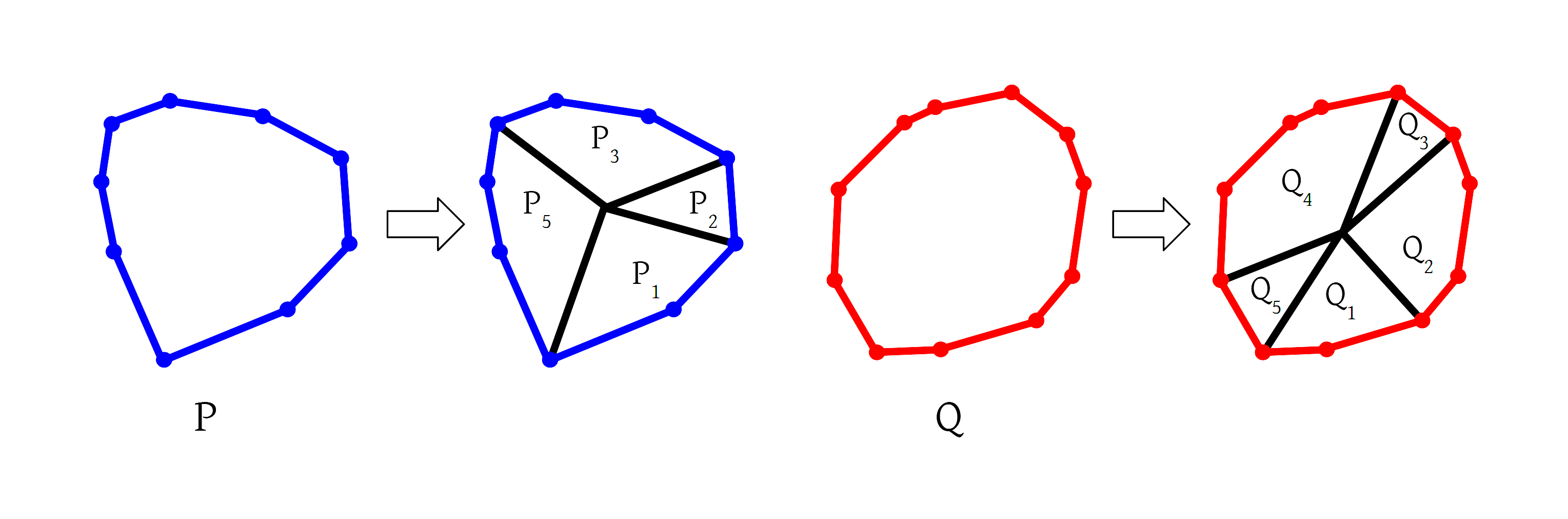}
    \caption{A partition of example convex polygons $P$ and $Q$ into blocks, with $b = 4$. Note that the spatial \emph{sizes} of blocks may vary greatly, and the number of edges in each $P_i$ and $Q_j$ may differ. In fact, for some indices $i$, one of $P_i$ or $Q_i$ may contain just the origin.}
    \label{fig:blocks}
\end{figure}

\paragraph{A Data Structure Using Blocks.} %\label{sec:blockstructure}
As we hone in on the translation $t^*$ maximizing $\textnormal{Area}((P+t^*) \cap Q)$, we want some way to decrease the time required to compute $\textnormal{Area}((P+t) \cap Q)$ for different query translations~$t$. 
%complexity of $P$ and $Q$ as we hone in on the vector $t^*$ which maximizes $\textnormal{Area}((P+t^*) \cap Q)$.
Instead of simplifying $P$ or $Q$ directly, we maintain and refine a specific data structure for $P$ and $Q$ that allows us to more efficiently search for optimal placements.

\begin{definition}[$(\mu, b, \mathcal{T})$-Block Structure]
A \emph{$(\mu, b, \mathcal{T})$-block structure} is a data structure containing the following:
    \begin{enumerate}
        \item Access to the vertices %and oriented edges
        of $P$ and $Q$ in counterclockwise order with $O(1)$ query time, and access to the area prefix sums for $P$ and $Q$ with $O(1)$ query time. 
        \item A block size parameter $b$, which determines blocks $P_1, \ldots, P_{\lceil N/b\rceil}$ and $Q_1, \ldots, Q_{\lceil N/b\rceil}$. The blocks are not stored explicitly; instead, the algorithm has access in $O(1)$ query time to the sorted multiset
        \[U = \{\theta_P(e) : e \textnormal{ is an edge of $P$}\} \,\cup\, \{\theta_Q(e) : e \textnormal{ is an edge of $Q$}\},\]
        and access in $O(1)$ query time to the edges of $P$ and $Q$ sorted by their angle ranges.
        \item A region $\mathcal{T} \subset \mathbb{R}^2$ that is either (i) a triangle, (ii) a line segment, or (iii) a point.
        \item A partition of $S := [\lceil N/b\rceil]$ into subsets $S_{\textnormal{Good}}$ and $S_{\textnormal{Bad}}$, with the requirement $\vert S_{\textnormal{Bad}}\vert \leq \mu$. %Throughout, we assume that $S_{\textnormal{Good}}$ and $S_{\textnormal{Bad}}$ are sorted, since we can apply radix sort in time $O(\vert S_{\textnormal{Good}}\vert + \vert S_{\textnormal{Bad}}\vert)$.
        \item For every $i \in S_{\textnormal{Good}}$, access in $O(1)$ query time to a %single
        quadratic function $f_i$ such that
        \[f_i(t) = \textnormal{Area}((P_i+t)\cap Q_i) \textnormal{ for all } t \in \mathcal{T}.\]
        %with the line-segment convention from Lemma~\ref{lem:existsquadratic} when $\mathcal{T}$ is one-dimensional.
    \end{enumerate}
\end{definition}

Observe that using (1) and (2), the endpoints and orientation of any edge in a block $P_i$ or $Q_j$ can be found on-demand in $O(1)$ query time, without needing to construct the entire block partition. 
%By the oriented-edge convention above, each block structure can also return the endpoints and orientation of any edge of a block, or of a union of consecutive blocks, in $O(1)$ query time.
Our algorithm will maintain several different block structures. All of them will reference a common data structure for the vertices/edges of $P$, a common data structure for the vertices/edges of $Q$, a common set of area prefix sums, and a common sorted multiset $U$.

\paragraph{Sketch of the Algorithm for Problem \ref{prob:overlap}.}
%We now sketch how block structures are used to solve Problem \ref{prob:overlap} in linear time.
At the onset of our linear time algorithm, we start with a $(\mu_0, b_0, \mathcal{T}_0)$-block structure where $\mu_0 = \lceil N/b_0\rceil$, $b_0 = 2$, and $\mathcal{T}_0$ is any triangle containing a superset of the translations $t$ that have $P+t$ intersecting $Q$ (and hence must contain the optimal translation). The partition of $S$ will simply be $S_{\textnormal{Bad}} = [\lceil N/2\rceil]$ and $S_{\textnormal{Good}} = \emptyset$. This corresponds to having no useful information about the objective function $\textnormal{Area}((P+t) \cap Q)$. To initialize the $(\mu_0, b_0, \mathcal{T}_0)$-block structure, essentially all we need to do is: (i) compute the area prefix sums for $P$ and $Q$ once and for all, (ii) sort $U$ once and for all (in linear time), and (iii) identify any triangle $\mathcal{T}_0$ such that, for all $t \in \mathbb{R}^2$, if $(P+t)\cap Q$ is nonempty then $t \in \mathcal{T}_0$.

As the algorithm progresses, we shrink $\mathcal{T}_0$, producing a sequence of smaller and smaller triangles $\mathcal{T}_0 \supseteq \mathcal{T}_1 \supseteq \mathcal{T}_2 \supseteq \mathcal{T}_3 \supseteq \ldots$, each of which is guaranteed to contain the translation $t^* \in \mathbb{R}^2$ maximizing $\textnormal{Area}((P+t^*) \cap Q)$. %Each of the triangles will need to satisfy certain requirements, which we sketch later in the overview.
After making each triangle $\mathcal{T}_\ell$, we immediately construct a $(\mu_\ell, b_\ell, \mathcal{T}_\ell)$-block structure, where $\mu_\ell$ gets significantly smaller as $\ell$ increases, and $b_\ell$ gets bigger as $\ell$ increases (see Section \ref{sec:cascade}). Intuitively, each block structure records the information about the objective function that we have learned so far.

Our approach for finding the sequence of triangles uses recursion in the dimension in a manner reminiscent of low-dimensional linear programming \cite{Megiddo83a,Dyer84,Megiddo84}, as detailed in Section \ref{sec:linalg}. We identify each of the next triangles $\mathcal{T}_{\ell+1} \subseteq\mathcal{T}_\ell$ by recursively solving a small number of \emph{line segment optimizations}, where each line segment is contained within $\mathcal{T}_\ell$. Each line segment subproblem is again solved recursively, this time using point optimizations (i.e. finding the area of intersection for a few given translations), where each point is contained within the line segment. Note that each of the line segment subproblems will in general \emph{not} contain the global optimum, and similarly for the point subproblems.

For the algorithm to be efficient, we need to construct the entire sequence of triangles $\mathcal{T}_0 \supseteq \mathcal{T}_1 \supseteq \mathcal{T}_2 \supseteq \mathcal{T}_3 \supseteq \ldots$ in linear time. %, along with the associated $(\mu_\ell, b_\ell, \mathcal{T}_\ell)$-block structures, in linear time. %When constructing $\mathcal{T}_1$ from $\mathcal{T}_0$, we haven't learned much about the convex polygons $P$ and $Q$, so we have little hope of performing this step in sublinear time.
By the time we are constructing $\mathcal{T}_{\ell+1} \subseteq\mathcal{T}_\ell$, observe that:
\begin{enumerate}
    \item We already have access to a $(\mu_\ell, b_\ell, \mathcal{T}_\ell)$-block structure.
    \item The only subproblems that are not solved by recursion are the point queries, which amount to reporting $\textnormal{Area}((P+\hat{t}) \cap Q)$ for different translations $\hat{t} \in \mathcal{T}_\ell$.
\end{enumerate}
When $\ell$ is small, the block structure is not very useful, and (for example) we can default to using an algorithm by O'Rourke, Chien, Olson, and Naddor \cite{o1982new} to answer each of the point queries in $O(N)$ time.\footnote{Our actual algorithm for the point oracle will be the same regardless of $\mu_\ell$ and $b_\ell$, and we will not use the algorithm in \cite{o1982new}.} But when $\ell$ is superconstant, we set the recursive parameters so that $\mu_\ell$ will be sublinear in $\lceil N/b_\ell\rceil$, and $b_\ell$ will be superconstant. This means that the $(\mu_\ell, b_\ell, \mathcal{T}_\ell)$-block structure will have quadratic summary functions $f_i(t)$ for nearly all of the block pairs $(P_i, Q_i)$. We show in Section \ref{sec:t0} that, in this case, the summary functions can be used to answer point queries in \emph{sublinear} time. In the final recurrence, the time to construct each of the triangles $\mathcal{T}_\ell$ will decrease geometrically, so the overall running time is indeed linear.

% . In other words, for nearly all of the block pairs $(P_i, Q_i)$, we can compute $\textnormal{Area}((P_i+t')\cap Q_i)$ in $O(1)$ time using the quadratic summary functions.

% The purpose of maintaining block structures for each triangle is exactly to accomplish this goal.

% This is whe $(\mu_\ell, b_\ell, \mathcal{T}_\ell)$-block structure for each triangle $\mathcal{T}_\ell$ is pre

% in the sequence, where $\mu_\ell$ gets significantly smaller as $\ell$ increases, and $b_\ell$ gets larger as $\ell$ increases. To see why this is helpful, fix a specific value of $\ell$, and consider finding $\mathcal{T}_{\ell+1} \subset \mathcal{T}_\ell$ with the help of the $(\mu_\ell, b_\ell, \mathcal{T}_\ell)$-block structure. When constructing $\mathcal{T}_{\ell+1}$, all of the recursive subproblems will be for regions contained within $\mathcal{T}_\ell$. This means that we can take advantage of the quadratic summary functions in the $(\mu_\ell, b_\ell, \mathcal{T}_\ell)$-block structure %stored in the $(\mu_\ell, b_\ell, \mathcal{T}_\ell)$-block structure and
% to solve all of the subproblems more efficiently. %(see below). 
% %can leverage the information stored in the $(\mu_i, b_i, \mathcal{T}_i)$-block structure. This is because all of the recursive subproblems will be for regions contained within $\mathcal{T}_i$, so the quadratic summary functions stored in item (5) of the block structure are valid for all of the subproblems.

\section{Using Block Structures for Point Queries} \label{sec:t0}

Below, we formalize the problem of finding the optimal translation within a block structure's region $\mathcal{T}$.

\begin{restatable}[\textsc{MaxRegion}]{problem}{MaxRegionProblem}
    Given a $(\mu, b, \mathcal{T})$-block structure, %for $P$ and $Q$,
    find \[\max_{t \in \mathcal{T}} \textnormal{Area}((P+t)\cap Q)\] along with the translation $t^* \in \mathbb{R}^2$ realizing this maximum.
\end{restatable}

Throughout, we use $T_{d\textsc{-Max}}(N, b, \mu)$ to denote the expected running time of the fastest algorithm for \textsc{MaxRegion}, where $d$ is the dimension of $\mathcal{T}$.

As mentioned in Section \ref{sec:blockstructure}, the instances of \textsc{MaxRegion} where $\mathcal{T}$ is a triangle or line segment will be solved recursively. For example, to solve \textsc{MaxRegion} with respect to a $(\mu, b, \mathcal{T})$-block structure where $\mathcal{T}$ is a triangle, we will solve a small number of \textsc{MaxRegion} instances for line segments $\mathcal{T}' \subseteq\mathcal{T}$, and then recursively solve \textsc{MaxRegion} for a smaller triangle (with a ``better'' block structure). Due to the containment property, the starting $(\mu, b, \mathcal{T})$-block structure \emph{also} serves as a $(\mu, b, \mathcal{T}')$-block structure for each line segment $\mathcal{T}'$, so the time to solve the line segment subproblems is at most $T_{1\textsc{-Max}}(N, b, \mu)$. A similar line of reasoning shows that \textsc{MaxRegion} with respect to a $(\mu, b, \mathcal{T})$-block structure where $\mathcal{T}$ is a line segment gives rise to point queries whose time complexity is at most $T_{0\textsc{-Max}}(N, b, \mu)$.

In this section, we show how to use block structures to answer the point queries efficiently:

% We show in Section \ref{sec:linalg} that, when $\mathcal{T}$ is a triangle or a line segment, \textsc{MaxRegion} can be solved via recursion.

% Roughly speaking, \textsc{MaxRegion} with $\mathcal{T}$ being a triangle reduces to a few instances of \textsc{MaxRegion} with $\mathcal{T}$ being a line segment, and \textsc{MaxRegion} with $\mathcal{T}$ being a line segment reduces to a few instances of \textsc{MaxRegion} with $\mathcal{T}$ being just a single point. The purpose in having a block structure is exactly to make the point case more efficient:

% and line segment optimization reduces to a few instances of point optimization. In this section, we given an efficient algorithm for the case that $\mathcal{T}$ is just a single point

% In this section, we give an algorithm for the point optimization case. In other words, $\mathcal{T}$ 

% In this section, we prove the Point Oracle Lemma, restated below:

\begin{restatable}[Point Oracle]{lemma}{PointOracleLemma}
\label{lem:T0}
    For all $2 \leq b \leq N$, for all $\mu$, %we have:
    \[T_{0\textsc{-Max}}(N, b, \mu) \leq O\left(N \cdot \frac{\log b}{b} + \mu b^2\right).\]
    That is, given a $(\mu, b, \mathcal{T})$-block structure where $\mathcal{T}$ is a single point $t^*$, we can compute $\textnormal{Area}((P+t^*)\cap Q)$ within the above time bound.
\end{restatable}

% Observe that in the point case

% %\noindent
% %\textbf{Lemma \ref{lem:T0}} (Point Oracle, Restated).
% %    \emph{For all $2 \leq b \leq N$, for all $\mu$,
% %    \[T_{0\textsc{-Max}}(N, b, \mu) = O\left(N \cdot \frac{\log^2 b}{b} + \mu b\right).\]}

% In other words, we need to use a given $(\mu, b, \mathcal{T})$-block structure (where $\mathcal{T}$ is just a single point $t^*$) to calculate $\textnormal{Area}((P+t^*) \cap Q)$.

Observe that we can decompose the area function as:
\[\textnormal{Area}((P+t^*)\cap Q) \:=\: \sum_{i \in [\lceil N/b\rceil]} \textnormal{Area}((P_i+t^*)\cap Q_i) + \sum_{\substack{(i, j) \in [\lceil N/b\rceil] \times [\lceil N/b\rceil] \\ i \neq j}} \textnormal{Area}((P_i+t^*)\cap Q_j),\]
where $P_1, \ldots, P_{\lceil N/b\rceil}$ and $Q_1, \ldots, Q_{\lceil N/b\rceil}$ are the blocks of $P$ and $Q$ determined by parameter~$b$.
To prove the lemma, it will suffice to compute each of the two summations on the right-hand side within the required time bound.
%\textcolor{red}{Probably good to use three paragraph headings: computing the first summation, intersecting unions of blocks, computing the second summation}

\paragraph{Computing the First Summation.} We can find $\sum_i \textnormal{Area}((P_i+t^*)\cap Q_i)$ by just reading the lists $S_{\textnormal{Good}}$ and $S_{\textnormal{Bad}}$ that are provided.

\begin{lemma} \label{lem:firstsummation}
    There is an algorithm running in time $O(N/b + \mu b^2)$ that takes as input a $(\mu, b, \mathcal{T})$-block structure, % where $2 \leq b \leq N^{o(1)}$ and
    where $\mathcal{T}$ consists of just a single point $t^*$, and outputs
    \[\sum_{i \in [\lceil N/b\rceil]} \textnormal{Area}((P_i+t^*)\cap Q_i).\]
    Here $P_1, \ldots, P_{\lceil N/b\rceil}$ and $Q_1, \ldots, Q_{\lceil N/b\rceil}$ are the blocks of $P$ and $Q$ determined by parameter~$b$.
\end{lemma}

\begin{proof}
    By definition of the block structure, we have a partition of $S = [\lceil N/b\rceil]$ into two subsets $S_{\textnormal{Good}}$ and $S_{\textnormal{Bad}}$, such that $\vert S_{\textnormal{Bad}}\vert \leq \mu$, and for each $i \in S_{\textnormal{Good}}$ we have a %single
    quadratic function $f_i(t)$ such that $f_i(t^*) = \textnormal{Area}((P_i + t^*) \cap Q_i)$. We can thus compute
    \[\sum_{i \in S_{\textnormal{Good}}}\textnormal{Area}((P_i + t^*) \cap Q_i)\]
    in time $O(N/b)$ by simply evaluating each $f_i(t^*)$. Since $\vert S_{\textnormal{Bad}}\vert \leq \mu$, we can compute
    \[\sum_{i \in S_{\textnormal{Bad}}}\textnormal{Area}((P_i + t^*) \cap Q_i)\]
    by calculating the area for each $(P_i + t^*) \cap Q_i$ directly. Because each block is a simple polygon with $O(b)$ edges, the time required to compute $\textnormal{Area}((P_i+t^*) \cap Q_i)$ is $O(b^2)$ for each $i \in S_{\textnormal{Bad}}$, or $O(\mu b^2)$ over all $i \in S_{\textnormal{Bad}}$, by using standard algorithms for polygon intersection (see for example \cite{greiner1998efficient}).\footnote{Recall from Section \ref{sec:blockstructure} that we can efficiently compute each of the relevant blocks $P_i$ and $Q_i$ on-demand, including the orientations and ordering of their edges.}
    %Assume without loss of generality that for each $i$, the blocks $P_i$ and $Q_i$ are convex; otherwise, $P_i$ and $Q_i$ can be split into at most two convex pieces each, and it will suffice to find the area of intersection for each pair of pieces. Because the number of edges in each block is $O(b)$, the time required to compute $\textnormal{Area}((P_i+t^*) \cap Q_i)$ is $O(b)$ for each $i \in S_{\textnormal{Bad}}$, or $O(\mu b)$ over all $i \in S_{\textnormal{Bad}}$, using known algorithms for intersecting two convex polygons whose edges are in sorted order (one such algorithm is given in \cite{o1982new}).
    %If a block is not convex, we first split it into at most two convex polygonal sectors as in the proof of Lemma~\ref{lem:intersectarea}; the pieces have disjoint interiors and their total complexity is $O(b)$. 
\end{proof}

All that remains is to compute $\sum_{i \neq j} \textnormal{Area}((P_i+t^*)\cap Q_j)$. Each term amounts to computing the area of intersection for two blocks with \emph{disjoint} angle ranges, so in principle we could use binary searches to evaluate each term in $O(\mathop{\rm polylog}(b))$ time. However, this is too slow for our purposes, as there are a near quadratic number of terms. We will instead use a recursive decomposition of the sum. %decompose the summation so that we are intersecting \emph{unions} of blocks from $P$ and $Q$.

%Instead, we take advantage of additional structure to decompose the summation into $O(N/b)$ intersections of \emph{unions of blocks}, such that each intersection can still be computed via binary search.

\paragraph{Intersecting Unions of Blocks.} Before presenting our algorithm for the second summation, we give some tools for intersecting unions of blocks that have disjoint angle ranges. The following is due to Dobkin and Souvaine \cite{dobkin1991detecting}.

\begin{lemma} \label{lem:layer}
    Let $P$ and $Q$ be any two convex polygons, and let $X$ and $Y$ be any two connected polygonal chains whose edges are a subset of the edges of $P$ and $Q$, respectively. Assume that $\Lambda_P(X) \cap \Lambda_Q(Y) = \emptyset$. %The same statement applies if either chain is replaced by a translate.
    If $X$ and $Y$ intersect, then they intersect in either a single point, a line segment, or two points. Furthermore, if we have access to the vertices of $X$ and $Y$ in sorted order with $O(1)$ query time, then the intersection(s) %, together with all edge pairs of $X$ and $Y$ containing them,
    can be found in $O(\log\vert X\vert + \log \vert Y \vert)$ time.
\end{lemma}

Combining this with the area prefix sums stored in the $(\mu, b, \mathcal{T})$-block structure, we have the lemma below. We state the lemma in a slightly more general form than is required for the point oracle; in particular, we allow the region $\mathcal{T}$ in the $(\mu, b, \mathcal{T})$-block structure to be a triangle, line segment, or point. This is because we will also use the lemma in Section \ref{sec:cascade} when presenting our algorithm to create new block structures.

\begin{lemma} \label{lem:intersectarea}
    There is an algorithm which takes as input (i) a $(\mu, b, \mathcal{T})$-block structure, (ii) a translation $\hat{t} \in \mathbb{R}^2$,\footnote{The translation does not need to be contained within $\mathcal{T}$, although we will only use the lemma for translations $\hat{t} \in \mathcal{T}$.} and (iii) two intervals $\mathcal{I}_P, \mathcal{I}_Q \subseteq [\lceil N/b\rceil]$ such that $\mathcal{I}_P \cap \mathcal{I}_Q = \emptyset$, runs in time $O(\log\vert \mathcal{I}_P\vert + \log \vert \mathcal{I}_Q\vert + \log b)$, and behaves as follows. Define polygons
    \[A = \bigcup_{i \in \mathcal{I}_P} P_i \quad \text{and} \quad B = \bigcup_{j \in \mathcal{I}_Q} Q_j,\]
    where $P_1, \ldots, P_{\lceil N/b\rceil}$ and $Q_1, \ldots, Q_{\lceil N/b\rceil}$ are the blocks of $P$ and $Q$ determined by parameter~$b$. The algorithm outputs $\textnormal{Area}((A+\hat{t}) \cap B)$, along with a list $L$ of size at most $O(1)$ containing every pair of edges $(e, e')$ from $A$ and $B$ such that $\hat{t} \in \pi(e, e')$.
    % \[\textnormal{Area}\left(\left(\bigcup_{i \in \mathcal{I}_P} P_i + \hat{t}\right) \cap \left(\bigcup_{j \in \mathcal{I}_Q} Q_j\right)\right),\]
    % along with a list $L$
    % Here $P_1, \ldots, P_{\lceil N/b\rceil}$ and $Q_1, \ldots, Q_{\lceil N/b\rceil}$ are the blocks of $P$ and $Q$ determined by parameter~$b$.
\end{lemma}

\begin{proof}
    Notice that $A$ is a union of consecutive blocks from $P$, and $B$ is a union of consecutive blocks from~$Q$.
    Assume without loss of generality that both $A$ and $B$ are convex polygons. Otherwise, we can partition $A$ and $B$ into at most two convex pieces each, using a line segment from the origin to an appropriate point on the boundary of each polygon. (See Figure \ref{fig:SplitA}.) To get the final area of intersection, we compute the area of intersection for each pair of convex pieces and add the results. To get the final list $L$, we take the union of all lists for each pair of convex pieces, and delete any edge pairs containing an edge that was added when partitioning $A$ and $B$.

\begin{figure}
    \centering
    \includegraphics[width=\textwidth]{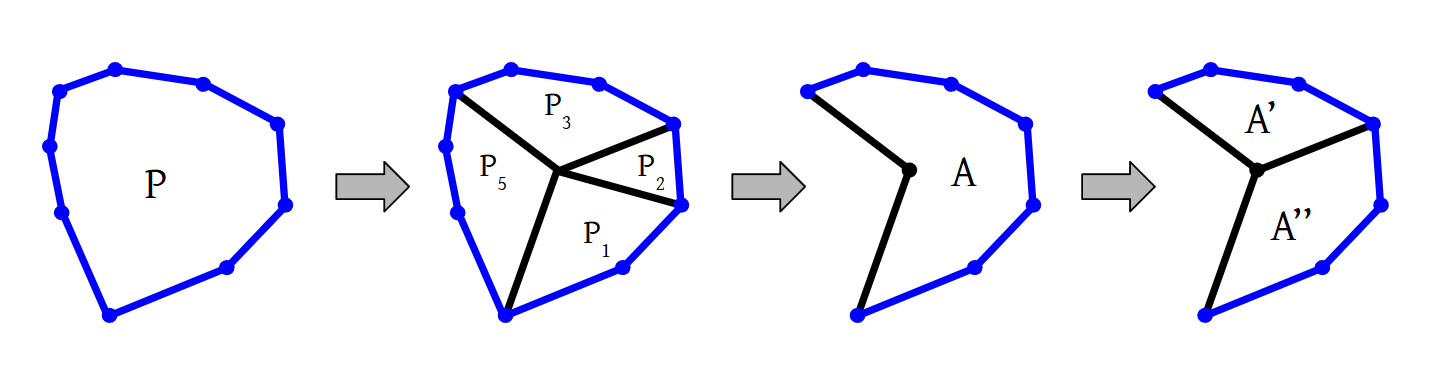}
    \caption{An example polygon $A = \bigcup_{i \in \mathcal{I}_P} P_i$, where $\mathcal{I}_P = \{1,2,3\}$. Notice that the origin is a reflex vertex. We split $A$ into two convex polygons $A'$ and $A''$ using a line segment from the origin to an appropriate vertex on the boundary of $A$.}
    \label{fig:SplitA}
\end{figure}

    %Thus $A$ is bounded by a chain on the boundary of $P$ together with two line segments connecting the endpoints of the chain to the origin, and similarly for $B$.

%Notice that $A$ is a union of consecutive blocks from $P$, and $B$ is a union of consecutive blocks from~$Q$. Assume without loss of generality that $A$ and $B$ are both convex polygons. Otherwise, we can partition $A$ and $B$ into at most two convex polygons each, and then compute for each pair of pieces the area of intersection and the corresponding list $L$. (\textcolor{red}{See Figure XXX.}) %To get the final area of intersection, we just add the area of intersection for each pair of pieces.
    %To get the final list $L$, we just take the union of all lists identified and delete any edge pairs containing an edge that was added when partitioning $A$ and $B$.

    %If either $A$ or $B$ is not convex, adding one line segment from the origin to an appropriate point on the corresponding subchain partitions it into two convex pieces. This means that we can assume without loss of generality that $A$ and $B$ are both convex polygons. Otherwise, we (i)

    By definition of our blocking scheme, and by the assumption that $\mathcal{I}_P \cap \mathcal{I}_Q = \emptyset$, we know that the edges that $A$ shares with $P$ and the edges that $B$ shares with $Q$ form two convex polygonal chains $X$ and $Y$ that have disjoint angle ranges. Stated more formally, we set $X = E_P(A)$ and $Y = E_Q(B)$, and we have $\Lambda_P(X) \cap \Lambda_Q(Y) = \emptyset$. The set of edges in $X$ contains all but two edges of $A$; let the remaining edges be $e_1$ and $e_2$. Similarly, the set of edges in $Y$ contains all but two edges of $B$; let the remaining edges be $e_1'$ and $e_2'$.

    The most difficult part in computing $\textnormal{Area}((A+\hat{t}) \cap B)$ is to identify the intersection(s) between the boundary of $A+\hat{t}$ and the boundary of $B$. First we apply Lemma \ref{lem:layer} to find the intersection(s) between $X+\hat{t}$ and $Y$, which takes time $O(\log\vert X\vert + \log \vert Y\vert) = O(\log\vert \mathcal{I}_P\vert + \log \vert \mathcal{I}_Q\vert + \log b)$. (Observe that the angle range conditions are invariant under translation.) As per the lemma, we know that the total number of intersections is at most two. All of the remaining intersections must involve at least one of the edges $e_1, e_2, e_1', e_2'$. Using that $X$ and $Y$ are convex polygonal chains, the number of additional intersections is at most $O(1)$, and we can find them in time $O(\log\vert X\vert + \log \vert Y\vert) = O(\log\vert \mathcal{I}_P\vert + \log \vert \mathcal{I}_Q\vert + \log b)$ via known algorithms to intersect a line segment with a convex polygonal chain (see for example \cite{dobkin1991detecting}). Let $L$ be the set of all edge pairs corresponding to an intersection. More precisely, $L$ contains every pair $(e, e')$ such that $e$ is an edge of $A$, $e'$ is an edge of $B$, and $\hat{t} \in \pi(e, e')$. In degenerate cases where the intersection is a shared endpoint or line segment, we include all incident edge pairs; this changes $\vert L\vert$ by only a constant factor. Observe that $L$ can be computed in $O(1)$ time by examining the $O(1)$ intersections. %\footnote{Suppose that $X$ is defined by three vertices $x_1, x_2, x_3 \in \mathbb{R}^2$, and that $X$ has an edge $e_4$ connecting $x_1$ and $x_2$ along with an edge $e_5$ connecting $x_2$ and $x_3$. Suppose further that $Y$ is defined by $y_1, y_2, y_3 \in \mathbb{R}^2$ and has an edge $e_4'$ connecting $y_1$ and $y_2$, along with an edge $e_5'$ connecting $y_2$ and $y_3$. If it happens that $x_2 + \hat{t} = y_2$, then $X+\hat{t}$ and $Y$ could intersect just once, but we count all four edge pairs $(e_4, e_4')$, $(e_4, e_5')$, $(e_5, e_4')$, $(e_5, e_5')$ as containing the intersection.}

    We now describe how to compute $\textnormal{Area}((A+\hat{t}) \cap B)$. First suppose that $L = \emptyset$. This means that exactly one of the following holds:
    \begin{enumerate}
        \item $(A+\hat{t}) \cap B = A+\hat{t}$,
        \item $(A+\hat{t}) \cap B = B$, or
        \item $(A+\hat{t}) \cap B = \emptyset$.
    \end{enumerate}
    We can determine which case holds by (i) checking whether a vertex of $(A+\hat{t})$ is contained within $B$, and then (ii) checking whether a vertex of $B$ is contained within $(A+\hat{t})$. The time is $O(\log\vert \mathcal{I}_P\vert + \log \vert \mathcal{I}_Q\vert + \log b)$ because $A$ and $B$ are convex polygons. %, and because we can recover their edges on-demand in $O(1)$ query time using the $(\mu, b, \mathcal{T})$-block structure.
    %using standard algorithms to locate a point with respect to a convex polygon. \textcolor{red}{Give reference?}
    After this, we can compute the area of intersection in $O(1)$ time, because the area prefix sums for $P$ and $Q$ allow us to determine both $\textnormal{Area}(A + \hat{t}) = \textnormal{Area}(A)$ and $\textnormal{Area}(B)$ in $O(1)$ query time.

    Now suppose that $L \neq \emptyset$. Using that $\vert L\vert = O(1)$, we know that the boundary of $(A+\hat{t}) \cap B$ can be partitioned into a constant number of convex polygonal chains, each of which is either on the boundary of $(A+\hat{t})$ or on the boundary of $B$. We can identify the chains in $O(1)$ time by examining each edge pair in~$L$. Let $X_1, \ldots, X_{O(1)}$ be the chains on the boundary of $(A+\hat{t})$, and let $Y_1, \ldots, Y_{O(1)}$ be the chains on the boundary of $B$. (Notice that each of the chains $X_i$ has already been translated by $\hat{t}$.)
    
    Because we assumed that $A$ and $B$ are convex polygons, the intersection $(A+\hat{t}) \cap B$ is also a convex polygon. We use this to decompose $(A+\hat{t}) \cap B$ as follows. For each $X_i$, let $A_i$ be its convex hull; for each $Y_j$, let $B_j$ be its convex hull. Define
    \[C = \left((A + \hat{t}) \cap B\right) \backslash \left(\left(\bigcup_i A_i\right) \cup \left(\bigcup_j B_j\right)\right).\]
    Because there are a constant number of $A_i$'s and $B_j$'s, $C$ has constant description complexity. $C$ can be computed in constant time using the list $L$.
    Because $(A+\hat{t}) \cap B$ is a convex polygon, the $A_i$'s, $B_j$'s, and $C$ constitute a partition of $(A + \hat{t}) \cap B$. (The pieces overlap at their boundaries, but the overlaps have zero area. See Figure \ref{fig:ABC} for an example.)

\begin{figure}
    \centering
    \includegraphics[width=0.6\textwidth]{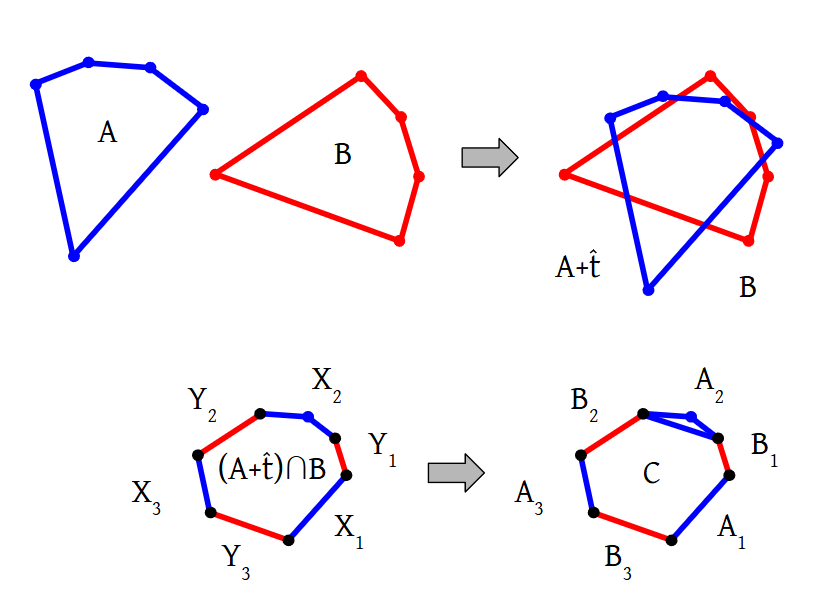}
    \caption{An example partition of $(A+\hat{t})\cap B$ into $A_i$'s, $B_j$'s, and $C$. Because this example is small, all of the $A_i$'s and $B_j$'s have zero area, except for $A_2$.}
    \label{fig:ABC}
\end{figure}
    
    Using the decomposition, we have
    \[\textnormal{Area}((A+\hat{t}) \cap B) = \sum_i\textnormal{Area}(A_i) + \sum_j\textnormal{Area}(B_j) + \textnormal{Area}(C).\]
    We can compute the first two summations using the area prefix sums for $P$ and $Q$, along with the inclusion-exclusion principle, %and the $O(1)$ endpoint corrections for the chords closing the chains (and for the translation by $\hat{t}$),
    in $O(1)$ time. %\textcolor{red}{Give more detail?}
    Computing $\textnormal{Area}(C)$ also takes $O(1)$ time.
\end{proof}

\paragraph{Computing the Second Summation.} Now we show how to compute $\sum_{i \neq j} \textnormal{Area}((P_i+t^*)\cap Q_j)$. %\textcolor{red}{Give one picture illustrating the decomposition.} %\textcolor{red}{The proof will be way easier if we just always intersect ``big versus small,'' as we will do later for the cascading subroutine. And just break each of the blocks / block groups into $O(1)$ pieces to ensure that we get convex polygons.}

\begin{lemma} \label{lem:secondsummation}
    There is an algorithm running in time $O(N \cdot \frac{\log b}{b})$ that takes as input a $(\mu, b, \mathcal{T})$-block structure, % where $2 \leq b \leq N^{o(1)}$ and
    where $\mathcal{T}$ consists of just a single point $t^*$, and outputs
    \[\sum_{\substack{(i, j) \in [\lceil N/b\rceil] \times [\lceil N/b\rceil] : \\ i \neq j}} \textnormal{Area}((P_i+t^*)\cap Q_j).\]
    Here $P_1, \ldots, P_{\lceil N/b\rceil}$ and $Q_1, \ldots, Q_{\lceil N/b\rceil}$ are the blocks of $P$ and $Q$ determined by parameter~$b$.
\end{lemma}

\begin{proof}
    We give a recursive algorithm for the problem. Given two integers $x, y \in [\lceil N/b\rceil]$ with $y \geq x$, define
\[W(x, y) := \{(i, j) \in [\lceil N/b\rceil] \times [\lceil N/b\rceil]: i \neq j,\:  x \leq i \leq y, x \leq j \leq y\}.\] Note that when $y = x$, we have $W(x, y) = \emptyset$. The recursive subproblems will amount to computing
\[\sum_{(i, j) \in W(x, y)}\!\!\textnormal{Area}((P_i + t^*) \cap Q_j)\]
for various integers $y > x$.
Let $T(y-x, b)$ denote the running time of the fastest algorithm to compute this sum.
%    \begin{align*}
%        & \sum_{(i, j) \in W(x, y)}\textnormal{Area}((P_i + t^*) \cap Q_j), \\
%        & 
%        \textnormal{ where}\ W(x, y) := \{(i, j) \in [N/b] \times [N/b]: i \neq j \textnormal{ and } x \leq i, j \leq y\}.
%    \end{align*}
    Observe that
    \[\sum_{\substack{(i, j) \in [\lceil N/b\rceil] \times [\lceil N/b\rceil] : \\ i \neq j}} \textnormal{Area}((P_i+t^*)\cap Q_j) = \sum_{(i, j) \in W(1, \lceil N/b\rceil)}\textnormal{Area}((P_i + t^*) \cap Q_j).\]
    In other words, the lemma amounts to upper bounding $T(\lceil N/b\rceil - 1, b)$. In the following, we show that $T(z, b) \leq 2T(\lfloor z/2\rfloor, b) + O(\log z + \log b)$ when $z \geq 2$, and $T(1, b) = O(\log b)$; solving the recurrence gives the desired upper bound $T(\lceil N/b\rceil - 1, b) = O(N \cdot \frac{\log b}{b})$.
    
    %$ = O(N \cdot \frac{\log b}{b})$. , or more generally, $T(z, b) = O(z\log b)$.

    Fix any positive integers $x, y \in [\lceil N/b\rceil]$ with $y > x$. Let $g := \lceil \frac{x+y}{2}\rceil$ and $z := y-x$. We decompose the target summation as follows:
    \begin{align*}
        \sum_{(i, j) \in W(x, y)}&\textnormal{Area}((P_i + t^*) \cap Q_j)\ = \\
        & \sum_{(i, j) \in W(x, g-1)}\textnormal{Area}((P_i + t^*) \cap Q_j) \:+ \sum_{(i, j) \in W(g, y)}\textnormal{Area}((P_i + t^*) \cap Q_j) + {}\\
        & \sum_{x \leq i < g \leq j \leq y}\textnormal{Area}((P_i + t^*) \cap Q_j) \:+ \sum_{x \leq j < g \leq i \leq y}\textnormal{Area}((P_i + t^*) \cap Q_j).
    \end{align*}

    If $z > 2$, the first two summations can be evaluated in time at most $2 T(\lfloor \frac{y-x}{2}\rfloor, b) = 2T(\lfloor z/2\rfloor, b)$ via recursion.\footnote{The time bound follows from the observation that $y - \lceil \frac{x+y}{2}\rceil = \lfloor \frac{y-x}{2}\rfloor$ and $(\lceil \frac{x+y}{2}\rceil - 1) - x \leq \lfloor \frac{y-x}{2}\rfloor$.} If $z = 2$, then $x = g-1$ and
    \[\sum_{(i, j) \in W(x, g-1)}\textnormal{Area}((P_i + t^*) \cap Q_j) = 0,\]
    so only one recursive subproblem is required (although we still upper bound the time spent on the recursive subproblem by $2T(1, b)$). If $z = 1$, then $x = g-1$ and $g = y$, meaning
    \[\sum_{(i, j) \in W(x, g-1)}\textnormal{Area}((P_i + t^*) \cap Q_j) = \sum_{(i, j) \in W(g, y)}\textnormal{Area}((P_i + t^*) \cap Q_j) = 0,\]
    so no recursion is required. %(unless $x-y = 1$, in which case they are automatically zero).

    The last two summations can be evaluated directly, by invoking Lemma \ref{lem:intersectarea}. For example, to compute
    \[\sum_{x \leq i < g \leq j \leq y}\textnormal{Area}((P_i + t^*) \cap Q_j),\]
    we set $\mathcal{I}_P = [x, g-1]$, $\mathcal{I}_Q = [g, y]$, and $\hat{t} = t^*$. %Since the blocks form area-disjoint partitions of $P$ and $Q$, the area of $((\bigcup_{i\in \mathcal{I}_P}P_i)+t^*)\cap(\bigcup_{j\in \mathcal{I}_Q}Q_j)$ is exactly the desired double sum, up to boundary overlaps of area zero.
    The total time for each of the two invocations is $O(\log z + \log b)$.
\end{proof}

\section{Creating New Block Structures} \label{sec:cascade}

In this section, we describe how to efficiently construct a $(\mu', b', \mathcal{T}')$-block structure, assuming that we already have access to a $(\mu, b, \mathcal{T})$-block structure for some region $\mathcal{T} \supseteq \mathcal{T}'$. Before stating the problem formally, we discuss some requirements on $\mathcal{T}'$ that will enable us to construct the $(\mu', b', \mathcal{T}')$-block structure efficiently. Let $P_1', \ldots, P_{\lceil N/b'\rceil}'$ and $Q_1', \ldots, Q_{\lceil N/b'\rceil}'$ be the blocks of $P$ and $Q$ determined by parameter $b'$.

The main existential obstruction to finding a $(\mu', b', \mathcal{T}')$-block structure is that we require $S_{\textnormal{Bad}}$ to have size at most $\mu'$. We will usually have $\mu'$ being much smaller than $\lceil N/b'\rceil$, in which case nearly all of the indices $i \in [\lceil N/b'\rceil]$ must satisfy $i \not\in S_{\textnormal{Bad}}$ and hence $i \in S_{\textnormal{Good}}$. For each of these indices, it must be possible for us to find a %single
quadratic summary function $f_i(t)$ such that $f_i(t) = \textnormal{Area}((P_i'+t)\cap Q_i') \textnormal{ for all } t \in \mathcal{T}'$.

To reason about this, fix a single index $i \in [\lceil N/b'\rceil]$, and consider the block pair $(P_i', Q_i')$. By Lemma \ref{lem:existsquadratic}, we know that $\textnormal{Area}((P_i'+t) \cap Q_i')$ can be summarized if $(P_i'+t)\cap Q_i'$ has the same combinatorial description for all $t$ in the interior of $\mathcal{T}'$. So a basic requirement that would permit us to construct a $(\mu', b', \mathcal{T}')$-block structure is that all but $\mu'$ of the block pairs $(P_i', Q_i')$ satisfy this property. Unfortunately, this is a bit difficult to work with algorithmically. Towards finding a more useful requirement, let $\mathcal{A}$ be the arrangement of the line segments in $\Pi(P_i', Q_i')$. Notice that we can re-interpret Lemma \ref{lem:existsquadratic} in terms of $\mathcal{A}$ as follows:
\begin{enumerate}
    \item If $\mathcal{T}'$ is a triangle, then $\textnormal{Area}((P_i'+t) \cap Q_i')$ can be summarized if the interior of $\mathcal{T}'$ is contained within a single cell of $\mathcal{A}$.
    \item If $\mathcal{T}'$ is a line segment, then $\textnormal{Area}((P_i'+t) \cap Q_i')$ can be summarized if either (1) the interior of $\mathcal{T}'$ is contained within a single cell of $\mathcal{A}$, or (2) the interior of $\mathcal{T}'$ is contained within a single edge of $\mathcal{A}$.
\end{enumerate}
%On the other hand, if the interior of $\mathcal{T}'$ intersects multiple components of $\mathcal{A}$, then $\textnormal{Area}((P_i+t) \cap Q_i)$ can (essentially) only be represented by a \emph{piecewise} quadratic.
The same implications hold if $\mathcal{A}$ is the arrangement of lines in $\overleftrightarrow{\Pi}(P_i', Q_i')$. This means that we can construct a $(\mu', b', \mathcal{T}')$-block structure if $\mathcal{T}'$ is strictly intersected by at most $\mu'$ lines from the multiset
\begin{align} \label{eq:linesettt}
    \biguplus_{i \in [\lceil N/b'\rceil]}\overleftrightarrow{\Pi}(P_i', Q_i').
\end{align}
(Recall that if $\mathcal{T}'$ is a line segment, we say that $\mathcal{T}'$ is strictly intersected by a line $\mathcal{L}$ if $\mathcal{L}$ intersects the interior of $\mathcal{T}'$, \emph{and also} $\mathcal{L}$ does not contain $\mathcal{T}'$.)

%If $\mathcal{T}'$ is a line segment contained in one of these line extensions, that line is not counted as strictly intersecting $\mathcal{T}'$; this is harmless here for the following reason. If the line segment is wholly contained in the corresponding parallelogram, then the associated edge pair is present throughout $\mathcal{T}'$. If it is not wholly contained but has a nonempty relative-interior overlap with the parallelogram side, then an endpoint of that overlap lies on an adjacent side of the parallelogram, whose supporting line strictly intersects $\mathcal{T}'$.

This is more useful algorithmically because, in the final linear time algorithm (see Section \ref{sec:linalg}), it allows us to find an appropriate $\mathcal{T}'$ using cuttings as per Lemma \ref{cor:approxcut}. However, it is still slightly too weak for us to construct the new block structure efficiently. Our final requirement will reference block pairs $(P_i, Q_j)$ corresponding to the \emph{existing} $(\mu, b, \mathcal{T})$-block structure, resulting in a multiset of lines that constitutes a superset of (\ref{eq:linesettt}).

Below, we formally define the problem of constructing an improved block structure from an existing one:

\begin{restatable}[\textsc{Cascade}]{problem}{CascadeProblem}
%\begin{problem}[\textsc{Cascade}]
    Given a $(\mu, b, \mathcal{T})$-block structure, parameters $\mu'$ and $b'$ such that $b$ divides $b'$, and a triangle or line segment $\mathcal{T}' \subseteq \mathcal{T}$, either produce a $(\mu', b', \mathcal{T}')$-block structure, or report failure.
    Failure may only be reported if $\mathcal{T}'$ is strictly intersected by more than $\mu'$ lines from the multiset
    \begin{equation} \label{eq:setS}
        \mathcal{S}\ := \biguplus_{k \in [\lceil N/b'\rceil]} \left(\biguplus_{(i, j) \in C_k \times C_k}\overleftrightarrow{\Pi}(P_{i}, Q_{j})\right),
    \end{equation}
    where $P_1, \ldots, P_{\lceil N/b\rceil}$ and $Q_1, \ldots, Q_{\lceil N/b\rceil}$ are the blocks of $P$ and $Q$ determined by parameter~$b$, $C_k := [(k-1)\cdot \frac{b'}{b} + 1, k \cdot \frac{b'}{b}]$ for $1 \leq k < \lceil N/b'\rceil$, and $C_{\lceil N/b'\rceil} := [(\lceil N/b'\rceil-1)\cdot \frac{b'}{b} + 1, \lceil N/b\rceil]$. The lines in $\mathcal{S}$ are counted with multiplicity.
\end{restatable}

    % For each $k \in \{1, \ldots, \lceil N/b'\rceil-1\}$, define the set $C_k = \{(k-1)\cdot \frac{b'}{b} + 1, \ldots, k \cdot \frac{b'}{b}\}$, and for $k = \lceil N/b'\rceil$, define the set $C_k = \{(k-1)\cdot \frac{b'}{b} + 1, \ldots, \lceil N/b\rceil\}$.
    % the interior of $\mathcal{T}'$ intersects more than $\mu'/2$ lines from the set\footnote{Strictly speaking, the inner union should range over $i$ and $j$ with the additional requirement that $k \cdot \frac{b'}{b} + i \leq \lceil N/b\rceil$ and $k \cdot \frac{b'}{b} + j \leq \lceil N/b\rceil$, to avoid referencing non-existent blocks. We omit this detail for the sake of readability.}
    % \begin{equation} \label{eq:setS}
    %     \mathcal{S}\ := \bigcup_{k \in \{0, \ldots, \lceil N/b'\rceil-1\}} \left(\bigcup_{(i, j) \in [\frac{b'}{b}] \times [\frac{b'}{b}]}\overleftrightarrow{\Pi}(P_{k\cdot\frac{b'}{b} + i}, Q_{k\cdot\frac{b'}{b} + j})\right).
    % \end{equation}
    % Here $P_1, \ldots, P_{\lceil N/b\rceil}$ and $Q_1, \ldots, Q_{\lceil N/b\rceil}$ are the blocks of $P$ and $Q$ determined by $b$. %, and $P_1', \ldots, P_{\lceil N/b'\rceil}'$ and $Q_1', \ldots, Q_{\lceil N/b'\rceil}'$ are those determined by $b'$.
%\end{problem}

Let $T_{\textsc{Cascade}}(N, b, b', \mu)$ be the expected running time of the fastest algorithm for \textsc{Cascade}. Our upper bound on the time complexity won't depend on $\mu'$ or the dimensions of $\mathcal{T}$ and $\mathcal{T}'$.

Our goal is to prove the following:

%\begin{lemma}[Cascading Subroutine] 
\begin{restatable}[Cascading Subroutine]{lemma}{CascadingLemma}
\label{lem:Cascade}
    For all $2 \leq b, b' \leq N$ such that $b$ divides $b'$, for all $\mu$, %\textcolor{red}{In theory this should just become $\log b'$, not $\log b'\log b$. Propagate this change and also the change from $\log^2 b$ to $\log b$.}
    \[T_{\textsc{Cascade}}(N, b, b', \mu) \leq O\left(N \cdot \frac{\log b'}{b} + \mu b^2\right).\]
\end{restatable}
%\end{lemma}

\paragraph{Improved Tools for Intersecting Blocks.} Before giving our algorithm for \textsc{Cascade}, we need to strengthen Lemmas \ref{lem:layer} and \ref{lem:intersectarea} from Section \ref{sec:t0}. We start with a counterpart of Lemma \ref{lem:layer} that is useful for sets of translations, as opposed to just a specific translation. For a region $\mathcal{T}' \subset \mathbb{R}^2$ and chains $X$ and $Y$ with disjoint angle ranges, we want to efficiently determine if the intersection between $X+t$ and $Y$ has the same combinatorial description for all translations $t$ in the interior of $\mathcal{T}'$.

%All lemmas in this section are intended to apply whether $\mathcal{T}'$ is a triangle or a line segment. For a line segment, all interiors and containment tests are relative to the supporting line, and the ``vertices'' of $\mathcal{T}'$ are its two endpoints. Constant-time tests against a parallelogram $\pi(e,e')$ are performed on the closed parallelogram itself, not only on its supporting lines; thus a collinear overlap between a line-segment region and a side of $\pi(e,e')$ is detected unless the whole line segment is contained in $\pi(e,e')$. If a line-segment region lies on the supporting line of a side of $\pi(e,e')$ and partly, but not entirely, overlaps $\pi(e,e')$, then one endpoint of this overlap lies in the relative interior of $\mathcal{T}'$ and is also on an adjacent supporting line of $\pi(e,e')$; we charge this degenerate change to that adjacent line, which strictly intersects $\mathcal{T}'$ under the definition in the preliminaries.

\begin{lemma} \label{lem:minkowskiintersect}
    Let $P$ and $Q$ be any two convex polygons, and let $X$ and $Y$ be any two connected polygonal chains whose edges are a subset of the edges of $P$ and $Q$, respectively. Assume that $\Lambda_P(X) \cap \Lambda_Q(Y) = \emptyset$, and we have access to the vertices of $X$ and $Y$ in sorted order with $O(1)$ query time. Let $\mathcal{T}' \subset \mathbb{R}^2$ be any triangle or line segment. In $O(\log \vert X\vert + \log \vert Y\vert)$ time, we can determine whether $(X+t)\cap Y$ has the same combinatorial description for all $t$ in the interior of $\mathcal{T}'$.
\end{lemma}

\begin{proof}
    Recall that $(X+t)\cap Y$ has the same combinatorial description for all translations $t$ in the interior of $\mathcal{T}'$ if and only if, for every edge $e$ of $X$ and edge $e'$ of $Y$, one of the following holds:
    \begin{enumerate}
        \item $\mathcal{T}' \subseteq \pi(e, e')$.
        \item $\pi(e, e')$ is disjoint from the interior of $\mathcal{T}'$.
    \end{enumerate}
    
    Let the set of vertices of $\mathcal{T}'$ be $V(\mathcal{T}')$. For each $t \in V(\mathcal{T}')$, we know by Lemma~\ref{lem:layer} that $X+t$ and $Y$ can intersect at most twice, and we can find the intersection(s) in time $O(\log \vert X\vert + \log \vert Y\vert)$. %\footnote{There is a special case where at least one of the intersections occurs at a vertex of $X$ or a vertex of $Y$ This can be handled with slightly more work without increasing the asymptotic running time.}
    Let $L$ be the set of all edge pairs from $X$ and $Y$ containing an intersection under at least one of these translations. Define
    \[L_X = \{\textnormal{edge } e \textnormal{ of } X : \exists \textnormal{ edge } e' \textnormal{ of } Y \textnormal{ such that } (e, e') \in L\}\]
    and
    \[L_Y = \{\textnormal{edge } e' \textnormal{ of } Y : \exists \textnormal{ edge } e \textnormal{ of } X \textnormal{ such that } (e, e') \in L\}.\]
    Because $\vert V(\mathcal{T}') \vert = O(1)$, we have $\vert L\vert = O(1)$ and hence $\vert L_X\vert = O(1)$ and $\vert L_Y\vert = O(1)$.
    
    For each $(e, e') \in L_X \times L_Y$, verify that one of the following two cases hold:
    \begin{enumerate}
        \item $\mathcal{T}' \subseteq \pi(e, e')$.
        \item $\pi(e, e')$ is disjoint from the interior of $\mathcal{T}'$.
    \end{enumerate}
    This can be done via a direct computation in time $O(1)$. %; for a line-segment $\mathcal{T}'$, it is just a constant-size one-dimensional interval test along the supporting line.
    If some pair $(e, e') \in L_X \times L_Y$ satisfies neither of these cases, then we are done, so assume otherwise.

    Observe that every remaining edge pair $(e, e') \not\in L_X \times L_Y$ \emph{cannot} satisfy case (1) above. To see this, consider the contrapositive: $\mathcal{T}' \subseteq \pi(e, e')$ implies that $t \in \pi(e, e')$ for all $t \in V(\mathcal{T}')$, which means that $e$ would have been included in $L_X$ and $e'$ would have been included in $L_Y$, and hence $(e, e') \in L_X \times L_Y$. So for the remainder of the proof, it will suffice to check for an edge pair $(e, e')\not\in L_X \times L_Y$ such that $\pi(e, e')$ intersects the interior of $\mathcal{T}'$.

    %if there exists an extraneous edge pair $(e, e')$ satisfying neither of these conditions, then we must have either $e \not\in L_X$ or $e' \not\in L_Y$ (or both).

    %we know that, for every edge pair $(e, e')$ such that there exists a vertex $v \in V(\mathcal{T}')$ with $v \in \pi(e, e')$, every line segment in $\Pi(e, e')$ is disjoint from the interior of $\mathcal{T}'$. 

    %let $L'$ be the set of all edge pairs $(e, e') \in L_X \times L_Y$ satisfying case (1). Because $\mathcal{T}'$ is convex and each $\pi(e, e')$ is convex, we have that every translation $t$ in the interior of $\mathcal{T}'$ has $X+t$ and $Y$ intersecting at every edge pair in $L'$.

    %We are not quite done, because some extraneous edge pair $(e, e')$ could have $\Pi(e, e')$ intersecting the interior of $\mathcal{T}'$ without $\pi(e, e')$ containing a point in $V(\mathcal{T}')$. %, meaning that this new edge pair is potentially not in $L'$.
    Let $X^-$ be $X$ with its edges appearing in $L_X$ removed, and let $Y^-$ be $Y$ with its edges appearing in $L_Y$ removed. (See Figure \ref{fig:XY}.) By construction, for all translations $t \in \mathbb{R}^2$, we have that $X^- + t$ and $Y$ never intersect at an edge pair in $L_X \times L_Y$, and we have that $X + t$ and $Y^-$ never intersect at an edge pair in $L_X \times L_Y$. On the other hand, for every edge pair $(e, e') \not\in L_X \times L_Y$, either $e$ is an edge of $X^-$, or $e'$ is an edge of $Y^-$. So there exists an edge pair $(e, e')\not\in L_X \times L_Y$ such that $\pi(e, e')$ intersects the interior of $\mathcal{T}'$ if and only if at least one of the following is satisfied:
    \begin{enumerate}
        \item There exists a translation $t$ in the interior of $\mathcal{T}'$ such that $X^- + t$ intersects $Y$.
        \item There exists a translation $t$ in the interior of $\mathcal{T}'$ such that $X + t$ intersects $Y^-$.
    \end{enumerate}

\begin{figure}
    \centering
    \includegraphics[width=0.7\textwidth]{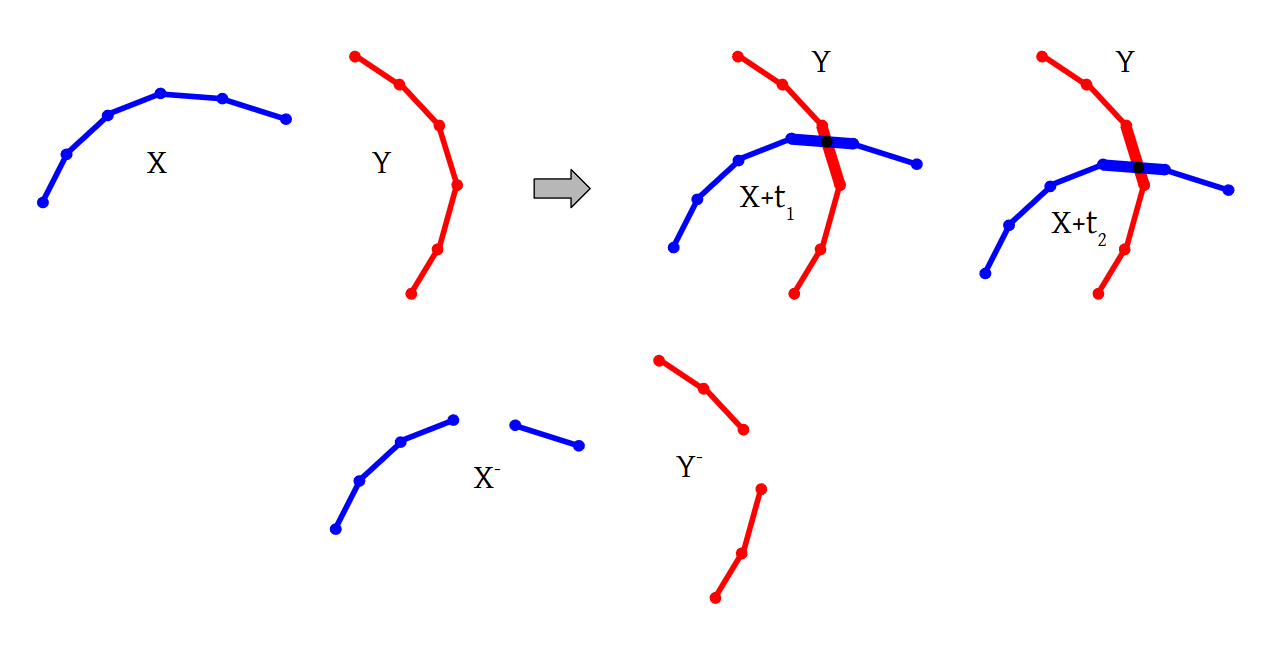}
    \caption{An example construction of $X^-$ and $Y^-$. In this case, we have $V(\mathcal{T}') = \{t_1, t_2\}$, so $\mathcal{T}'$ is a line segment. Observe that $X^-$ and $Y^-$ are not necessarily connected.}
    \label{fig:XY}
\end{figure}

    We describe how to check for case (1), as case (2) is symmetric. First decompose $X^-$ into $O(1)$ connected convex polygonal chains $X^-_1, \ldots, X^-_{O(1)}$ (these chains are delimited by the edges in $L_X$). Let $\mathcal{T}'^{-}$ be an infinitesimally smaller copy of $\mathcal{T}'$ that is strictly inscribed within $\mathcal{T}'$.\footnote{Infinitesimal shifts are a standard trick, and all of the algorithms in this paper can handle them with only minor modifications.} Observe that case (1) holds if and only if there exists an index $i$ such that the Minkowski sum $X^-_i + \mathcal{T}'^{-}$ intersects $Y$. In the rest of the proof, we describe how to check for an intersection involving a single $X^-_i$. Repeating an identical procedure for the other pieces only increases the running time by a constant factor.
    
    %If $\mathcal{T}'$ is a line segment and $X^-_i + \mathcal{T}'^{-}$ degenerates to a polygonal chain rather than a two-dimensional region, then the desired intersection test reduces to a constant number of chain-chain or segment-chain intersection tests and is no harder. Otherwise,
    While the entire Minkowski sum is not necessarily a convex polygon,
    its boundary is defined by $O(1)$ convex polygonal chains, and these can be identified in time $O(\log\vert X^-_i\vert)$ via binary searches on $X^-_i$. More concretely, $X^-_i + \mathcal{T}'^{-}$ is bounded by $O(1)$ polygonal chains $X^-_{i,1}, \ldots, X^-_{i, O(1)}$ that are translated copies of connected subchains of $X^-_i$, along with $O(1)$ additional edges $e_{i, 1}, \ldots, e_{i, O(1)}$.
    
    Observe that $X^-_i + \mathcal{T}'^{-}$ intersects $Y$ if and only if at least one of the following holds: %\footnote{Because $Y$ is a polygonal chain, if $X^-_i + \mathcal{T}'$ is contained within $Y$, then at least one of the first two cases must be satisfied.}
    \begin{enumerate}
        \item There exists an index $j$ such that $X^-_{i, j}$ intersects $Y$.
        \item There exists an index $j$ such that $e_{i, j}$ intersects $Y$.
        \item $Y$ is contained within $X^-_i + \mathcal{T}'^{-}$.
    \end{enumerate}
    To check for the first case, we apply Lemma \ref{lem:layer} for each choice of $j$. (As in Section \ref{sec:t0}, observe that the angle range conditions are invariant under translation.) We check for the second case using known algorithms to intersect a line segment with a convex polygonal chain (see for example \cite{dobkin1991detecting}), again iterating over each choice of $j$. To check for the third case, we orient an arbitrary vertex of $Y$ with respect to $X^-_i + \mathcal{T}'^{-}$, which takes $O(\log\vert X^-_i\vert)$ time because we already decomposed the boundary of $X^-_i + \mathcal{T}'^{-}$ into $O(1)$ convex polygonal chains. The total time to check for all three cases is $O(\log\vert X^-_i\vert + \log\vert Y\vert) \leq O(\log\vert X\vert + \log \vert Y\vert)$.
\end{proof}

Below we generalize Lemma \ref{lem:minkowskiintersect}, so that the algorithm applies to unions of blocks from $P$ and $Q$ instead of convex polygonal chains.

\begin{lemma} \label{lem:findanyintersect}
    There is an algorithm which takes as input (i) a $(\mu, b, \mathcal{T})$-block structure, (ii) a triangle or line segment $\mathcal{T}' \subseteq \mathcal{T}$, and (iii) two intervals $\mathcal{I}_P, \mathcal{I}_Q \subseteq [\lceil N/b\rceil]$ such that $\mathcal{I}_P \cap \mathcal{I}_Q = \emptyset$, runs in time $O(\log\vert \mathcal{I}_P\vert + \log \vert \mathcal{I}_Q\vert + \log b)$, and behaves as follows. Define simple polygons
    \[A = \bigcup_{i \in \mathcal{I}_P} P_i \quad \text{and} \quad B = \bigcup_{j \in \mathcal{I}_Q} Q_j,\]
    where $P_1, \ldots, P_{\lceil N/b\rceil}$ and $Q_1, \ldots, Q_{\lceil N/b\rceil}$ are the blocks of $P$ and $Q$ determined by parameter~$b$. 
    The algorithm determines whether $(A+t)\cap B$ has the same combinatorial description for all translations $t$ in the interior of $\mathcal{T}'$.
\end{lemma}

\begin{proof}
    As in the proof of Lemma \ref{lem:intersectarea}, define two convex polygonal chains $X = E_P(A)$ and $Y = E_Q(B)$, and observe that $\mathcal{I}_P \cap \mathcal{I}_Q = \emptyset$ implies $\Lambda_P(X) \cap \Lambda_Q(Y) = \emptyset$.\footnote{In contrast to the proof of Lemma \ref{lem:intersectarea}, we don't need to assume that $A$ and $B$ are convex.} The set of edges in $X$ contains all but two edges of $A$; let the remaining edges be $e_1$ and $e_2$. Similarly, the set of edges in $Y$ contains all but two edges of $B$; let the remaining edges be $e_1'$ and $e_2'$.

    Apply Lemma \ref{lem:minkowskiintersect} to determine whether there exists an edge $e$ of $X$ and an edge $e'$ of $Y$ such that neither of the following holds:
    \begin{enumerate}
        \item $\mathcal{T}' \subseteq \pi(e, e')$.
        \item $\pi(e, e')$ is disjoint from the interior of $\mathcal{T}'$.
    \end{enumerate}
    The time required is $O(\log \vert X\vert + \log \vert Y \vert) = O(\log\vert \mathcal{I}_P\vert + \log \vert \mathcal{I}_Q\vert + \log b)$. If we detect an edge pair $(e, e')$ satisfying neither of these cases, then we are done, so assume otherwise. Every remaining pair $(e, e')$ is in one of the following sets:
    \begin{enumerate}
        \item $\{e_1, e_2\} \times \{e_1', e_2'\}$.
        \item $\{e_1\} \times \{e' : e' \textnormal{ is an edge of } Y\}$.
        \item $\{e_2\} \times \{e' : e' \textnormal{ is an edge of } Y\}$.
        \item $\{e : e \textnormal{ is an edge of } X\} \times \{e_1'\}$.
        \item $\{e : e \textnormal{ is an edge of } X\} \times \{e_2'\}$.
    \end{enumerate}

    Case (1) can be checked via a direct computation in time $O(1)$. We now describe how to check for case (2), as the remaining cases are symmetric. First pick an arbitrary auxiliary point $p_1$ that does not lie on the line extension of $e_1$, and let $Z$ be the convex hull of $e_1$ and $p_1$. The purpose in defining $Z$ is so that $e_1$ has an angle range $\theta_Z(e_1)$ defined with respect to a convex polygon. Partition $Y$ into two subchains $Y^-$ and $Y^+$, such that $\theta_Z(e_1) \not\in \Lambda_Q(Y^-)$ and $\theta_Z(e_1) \not\in \Lambda_Q(Y^+)$. (If there exists an edge ${e}'$ of $Y$ with $\theta_Z(e_1) = \theta_Q({e}')$, then remove ${e}'$ and check the pair $(e_1, e')$ directly.) The time required so far is $O(\log \vert Y\vert)$. Now simply apply Lemma \ref{lem:minkowskiintersect} to $e_1$ (interpreted as an edge on the auxiliary convex polygon $Z$) and each of $Y^-$ and $Y^+$.
    
    The total time for cases (2) - (5) is $O(\log \vert X\vert + \log \vert Y \vert) = O(\log\vert \mathcal{I}_P\vert + \log \vert \mathcal{I}_Q\vert + \log b)$. %If none of the checks above finds a violating edge pair, then every edge pair satisfies one of the two defining alternatives, and the algorithm reports that the combinatorial description is fixed over the interior of $\mathcal{T}'$. If any of these checks finds an edge pair violating the two alternatives above, the algorithm reports that the combinatorial description changes; otherwise, every possible edge pair has been checked and the algorithm reports that the combinatorial description is fixed on the interior of $\mathcal{T}'$.
\end{proof}

Now we give a generalized version of Lemma \ref{lem:intersectarea}. Suppose that we have a region $\mathcal{T}' \subset \mathbb{R}^2$, along with a union of blocks $A$ from $P$ and a union of blocks $B$ from $Q$, such that the angle ranges for $A$ and $B$ are disjoint. Also assume that the combinatorial description of $(A+t) \cap B$ is the same for all translations $t$ in the interior of $\mathcal{T}'$. We want to find a %single
quadratic function $f(t)$ such that $f(t) = \textnormal{Area}((A+t) \cap B)$ for all $t$ in $\mathcal{T}'$. %(Notice that if the quadratic function is valid for all translations in the interior of $\mathcal{T}'$, then it's also valid for translations on the boundary of $\mathcal{T}'$.) \textcolor{red}{Can improve the intuition given here. Just re-write this transition after re-defining stuff earlier.}
    
\begin{lemma} \label{lem:makequadratic}
    There is an algorithm which takes as input (i) a $(\mu, b, \mathcal{T})$-block structure, (ii) a triangle or line segment $\mathcal{T}' \subseteq \mathcal{T}$, and (iii) two intervals $\mathcal{I}_P, \mathcal{I}_Q \subseteq [\lceil N/b\rceil]$ such that $\mathcal{I}_P \cap \mathcal{I}_Q = \emptyset$, runs in time $O(\log\vert \mathcal{I}_P\vert + \log \vert \mathcal{I}_Q\vert + \log b)$, and behaves as follows. Define simple polygons
    \[A = \bigcup_{i \in \mathcal{I}_P} P_i \quad \text{and} \quad B = \bigcup_{j \in \mathcal{I}_Q} Q_j,\]
    where $P_1, \ldots, P_{\lceil N/b\rceil}$ and $Q_1, \ldots, Q_{\lceil N/b\rceil}$ are the blocks of $P$ and $Q$ determined by parameter~$b$. 
    The algorithm either outputs a %single
    quadratic function $f(t)$ such that $f(t) = \textnormal{Area}((A+t) \cap B)$ for all $t \in \mathcal{T}'$, or correctly reports that $(A+t)\cap B$ does not have the same combinatorial description for all translations $t$ in the interior of $\mathcal{T}'$.
    % the combinatorial description of $(A+t) \cap $ there exists an edge $e$ of $A$ and an edge $e'$ of $B$ such that neither of the following holds:
    % \begin{enumerate}
    %     \item $\mathcal{T}' \subseteq \pi(e, e')$.
    %     \item $\pi(e, e')$ is disjoint from the interior of $\mathcal{T}'$.
    % \end{enumerate}
\end{lemma}

\begin{proof}
    First apply Lemma \ref{lem:findanyintersect} to determine if $(A+t)\cap B$ has the same combinatorial description for all translations $t$ in the interior of $\mathcal{T}'$. If this is not the case, then we are done, so assume otherwise. Now let $\hat{t}$ be an arbitrary point in the interior of $\mathcal{T}'$, and invoke Lemma~\ref{lem:intersectarea} to compute $\textnormal{Area}((A+\hat{t}) \cap B)$, along with the list $L$ of size at most $O(1)$ containing every pair of edges on the boundaries of $A+\hat{t}$ and $B$ that intersect. If $L = \emptyset$, then we simply set $f(t) = \textnormal{Area}((A+\hat{t}) \cap B)$ and output this constant function, so assume otherwise. Because $\hat{t}$ lies in the interior of $\mathcal{T}'$ and the combinatorial description is fixed over this interior, every pair in $L$ is precisely a pair $(e,e')$ with $\mathcal{T}'\subseteq\pi(e,e')$, and no other pair has this property. The total time so far is $O(\log\vert \mathcal{I}_P\vert + \log \vert \mathcal{I}_Q\vert + \log b)$.

    As per Lemma \ref{lem:existsquadratic}, the gradient field of $\textnormal{Area}((A+t) \cap B)$, when $t$ is restricted to the interior of $\mathcal{T}'$, can be computed from $L$ in time $O(\vert L\vert)$. %\footnote{Because the gradient field is a linear function, it has constant description complexity.}
    Because $\vert L\vert = O(1)$, we can compute the gradient field in $O(1)$ time. Integrating this function, and using the reference point $\hat{t} \in \mathcal{T}'$ for which we know $\textnormal{Area}((A+\hat{t}) \cap B)$, we get a single function $f(t) = \textnormal{Area}((A+t) \cap B)$ that is quadratic in $t$. This requires $O(1)$ additional time.
\end{proof}
    
    % If we detect an intersection, then we are done. Otherwise, in the case that $\mathcal{T}'$ is a triangle, a line segment from $\Pi(X, Y)$ intersects the interior of $\mathcal{T}'$ if and only if at least one of the following occurs:
    % \begin{enumerate}
    %     \item There exists an index $k \in [2]$ such that the interior of $e_k + \mathcal{T}'$ contains a \emph{vertex} of $B$.
    %     \item There exists an index $k \in [2]$ such that a \emph{vertex} of $A$ is in the interior of $e_k - \mathcal{T}'$.
    % \end{enumerate}
    % To see why this is true, notice that case (1) is satisfied if and only if there is an edge $e_k$ and two different translations $t_1, t_2$ in the interior of $\mathcal{T}'$ such that $e_k+t_1$ and $e_k+t_2$ strictly intersect different edges of $B$. A similar argument holds for case (2). We can check for both cases in time $O(\log \vert X\vert + \log \vert Y \vert) = O(\log\vert \mathcal{I}_P\vert + \log \vert \mathcal{I}_Q\vert + \log b)$ using a procedure 

    % The Minkowski sums $e_i + \mathcal{T}'$ and $e_i' - \mathcal{T}'$ are defined by a constant number of line segments, so we can check for all of the remaining cases in $O(\log \vert X\vert + \log \vert Y \vert) = O(\log\vert \mathcal{I}_P\vert + \log \vert \mathcal{I}_Q\vert + \log b)$ by using known algorithms to intersect a line segment with a convex chain, and to orient a point with respect to a convex chain.

\paragraph{The Cascading Subroutine.}
We are now ready to give our algorithm for \textsc{Cascade}. Throughout, let $P_1, \ldots, P_{\lceil N/b\rceil}$ and $Q_1, \ldots, Q_{\lceil N/b\rceil}$ be the blocks of $P$ and $Q$ determined by parameter $b$, and let $P_1', \ldots, P_{\lceil N/b'\rceil}'$ and $Q_1', \ldots, Q_{\lceil N/b'\rceil}'$ be the blocks of $P$ and $Q$ determined by parameter $b'$.
%The main observation is that the quadratic summary functions stored in the $(\mu', b', \mathcal{T}')$-block structure only need to apply when $t \in \mathcal{T}'$, not the more general case that $t \in \mathcal{T}$, so many of the previously difficult blocks can now be summarized. We
The plan is to find a quadratic summary function for almost all indices $k \in [\lceil N/b'\rceil]$ in two steps: %\textcolor{red}{The proof of this lemma should be way easier, now that we've ``encapsulated'' the work into previous lemmas. But please be explicit about everything that the new block structure is just ``directly referencing'' from the old one, and also be explicit about how the blocks ``indeed correctly nest with each other'' in light of our new block decomposition definition.}

\begin{enumerate}
    \item Scan through the information provided in the input $(\mu, b, \mathcal{T})$-block structure, and use this to recover part of the function $\textnormal{Area}((P_k'+t) \cap Q_k')$.
    \item Recover the missing information for $\textnormal{Area}((P_k'+t) \cap Q_k')$ by intersecting pairs of unions of blocks with disjoint angle ranges.
    %We evaluate these using nested binary searches (Lemma~\ref{lem:layer} or \ref{lem:meshintersect}) and area prefix sums.
\end{enumerate}

\begin{proof}[Proof of Lemma \ref{lem:Cascade}.]
    %We start with some simple bookkeeping steps. %First, sort the sets $S_{\textnormal{Good}}$ and $S_{\textnormal{Bad}}$ from the existing $(\mu, b, \mathcal{T})$-block structure. This requires time $O(\vert S_{\textnormal{Good}}\vert + \vert S_{\textnormal{Bad}}\vert) = O(N/b)$ using radix sort.
    Observe that the existing $(\mu, b, \mathcal{T})$-block structure already provides the following:
    \begin{enumerate}
        \item Access to the vertices %and oriented edges
        of $P$ and $Q$ in counterclockwise order with $O(1)$ query time, and access to the area prefix sums for $P$ and $Q$ with $O(1)$ query time. 
        \item Access in $O(1)$ query time to the sorted multiset
        \[U = \{\theta_P(e) : e \textnormal{ is an edge of $P$}\} \,\cup\, \{\theta_Q(e) : e \textnormal{ is an edge of $Q$}\},\]
        and access in $O(1)$ query time to the edges of $P$ and $Q$ sorted by their angle ranges.
    \end{enumerate}
    The new $(\mu', b', \mathcal{T}')$-block structure will simply reference these existing data structures; the setup time is $O(1)$. All that remains is to compute the partition of $[\lceil N/b'\rceil]$ into sets $S_{\textnormal{Good}}'$ and $S_{\textnormal{Bad}}'$ and to find, for each $k \in S_{\textnormal{Good}}'$, a %single
    quadratic function $f_k'(t)$ such that $f_k'(t) = \textnormal{Area}((P_k'+t) \cap Q_k')$ for all $t \in \mathcal{T}'$.
    
    % The parameter $b$
    
    % \item Access to the vertices and oriented edges of $P$ and $Q$ in counterclockwise order with $O(1)$ query time, and access to the area prefix sums for $P$ and $Q$ with $O(1)$ query time. 
    %     \item A block size parameter $b$, which determines blocks $P_1, \ldots, P_{\lceil N/b\rceil}$ and $Q_1, \ldots, Q_{\lceil N/b\rceil}$. The blocks are not stored explicitly; instead, the algorithm has access in $O(1)$ query time to the sorted multiset
    %     \[U = \{\theta_P(e) : e \textnormal{ is an edge of $P$}\} \,\cup\, \{\theta_Q(e) : e \textnormal{ is an edge of $Q$}\}.\]

    % For the new $(\mu', b', \mathcal{T}')$-block structure, we can initialize the area prefix sums for $P$ and $Q$, and the sorted multiset
    % \[U = \{\theta_P(e) : e \textnormal{ is an edge of $P$}\} \,\cup\, \{\theta_Q(e) : e \textnormal{ is an edge of $Q$}\},\]
    % by simply referencing the copies that were already referenced by the $(\mu, b, \mathcal{T})$-block structure.

    Recall from the statement of \textsc{Cascade} that $C_k := [(k-1)\cdot \frac{b'}{b} + 1, k \cdot \frac{b'}{b}]$ for $1 \leq k < \lceil N/b'\rceil$, and $C_{\lceil N/b'\rceil} := [(\lceil N/b'\rceil-1)\cdot \frac{b'}{b} + 1, \lceil N/b\rceil]$. 
    %For each $k \in \{1, \ldots, \lceil N/b'\rceil-1\}$, define the set $C_k = \{(k-1)\cdot \frac{b'}{b} + 1, \ldots, k \cdot \frac{b'}{b}\}$, and for $k = \lceil N/b'\rceil$, define the set $C_k = \{(k-1)\cdot \frac{b'}{b} + 1, \ldots, \lceil N/b\rceil\}$.
    By definition of the blocking scheme, and because $b$ divides $b'$, we have that $P_k' = \bigcup_{i \in C_k}P_i$ and $Q_k' = \bigcup_{j \in C_k}Q_j$. 
    This gives the following decomposition:
    \begin{align*}
        \textnormal{Area}\left((P_k'+t) \cap Q_k'\right) & = \!\sum_{(i, j) \in C_k \times C_k}\!\!\!\!\textnormal{Area}((P_i+t) \cap Q_j) \\
        & = \sum_{i \in C_k}\textnormal{Area}((P_i+t) \cap Q_i)\ + \!\sum_{\substack{(i, j) \in C_k \times C_k:\\ i \neq j}}\!\!\!\!\textnormal{Area}((P_i+t) \cap Q_j).
        %& = \sum_{i \in C_k \cap S_{\textnormal{Good}}}\textnormal{Area}((P_i+t) \cap Q_i)\ + \sum_{i \in C_k \cap S_{\textnormal{Bad}}}\textnormal{Area}((P_i+t) \cap Q_i)\ + \!\sum_{\substack{(i, j) \in C_k \times C_k:\\ i \neq j}}\!\!\!\!\textnormal{Area}((P_i+t) \cap Q_j).
    \end{align*}

    We now describe how to compute $\sum_{i \in C_k}\textnormal{Area}((P_i+t) \cap Q_i)$. The algorithm can be viewed as a counterpart of the one in Lemma \ref{lem:firstsummation}.

    \begin{claim} \label{claim:cis}
        There is an algorithm running in time $O(N/b + \mu b^2)$ that behaves as follows. For each $k \in [\lceil N/b'\rceil]$, the algorithm will either output a %single
        quadratic function $f_k^{(1)}(t)$ such that
        \[f_k^{(1)}(t) = \sum_{i \in C_k}\textnormal{Area}((P_i+t) \cap Q_i)\]
        for all $t \in \mathcal{T}'$, or it will correctly determine that there exists an index $i \in C_k$ such that $(P_i+t)\cap Q_i$ does not have the same combinatorial description for all $t$ in the interior of $\mathcal{T}'$.
    \end{claim}

    \begin{proof}
        Initialize $f_k^{(1)}(t) = 0$ for all $k \in [\lceil N/b'\rceil]$. For each $i$, let $k(i)$ be the unique index such that $i \in C_k$. Scan through the list $S_{\textnormal{Good}}$ from the existing $(\mu, b, \mathcal{T})$-block structure. For each term $\textnormal{Area}((P_i+t) \cap Q_i)$ with $i \in S_{\textnormal{Good}}$, the existing $(\mu, b, \mathcal{T})$-block structure already stores a %single
        quadratic function $f_i(t)$ such that $f_i(t) = \textnormal{Area}((P_i+t) \cap Q_i)$ for all $t \in \mathcal{T}$, and that same function is valid on $\mathcal{T}' \subseteq \mathcal{T}$. Add each such $f_i(t)$ to $f_{k(i)}^{(1)}(t)$. The total time so far is $O(\vert S_{\textnormal{Good}}\vert) = O(N/b)$. %Because $S_{\textnormal{Good}}$ is sorted, reading these functions takes time $O(\frac{b'}{b})$ for each $k$, or $O(N/b)$ time over all choices of $k$.
    
        Now initialize sets $S_{\textnormal{Bad}}^+ = \emptyset$ and $S_{\textnormal{Bad}}^- = \emptyset$, and scan through the list $S_{\textnormal{Bad}}$ from the existing $(\mu, b, \mathcal{T})$-block structure. 
        %Globally, there are at most $\mu$ indices $i$ not in $S_{\textnormal{Good}}$. For each of these,
        Determine for each $i \in S_{\textnormal{Bad}}$ whether $(P_i+t)\cap Q_i$ has the same combinatorial description for all $t$ in the interior of $\mathcal{T}'$, by directly examining each line segment in $\Pi(P_i, Q_i)$ in time $O(b^2)$. Every $i$ satisfying this condition is added to $S_{\textnormal{Bad}}^+$, and every $i$ not satisfying this condition is added to $S_{\textnormal{Bad}}^-$. The total time is %$O(b^2)$ time for each $i$, and
        $O(\mu b^2)$ over all choices of $i \in S_{\textnormal{Bad}}$.
        %The endpoints and orientations of the relevant block edges are supplied by the oriented-edge convention above.

        For each $i \in S_{\textnormal{Bad}}^+$, we know by Lemma \ref{lem:existsquadratic} that $\textnormal{Area}((P_i + t) \cap Q_i)$ can be summarized by a %single
        quadratic function $g_i(t)$ for all $t \in \mathcal{T}'$. We now describe how to find $g_i(t)$ for each such index in time $O(b^2)$. Compute the set $L=\{(e,e'):\mathcal{T}'\subseteq\pi(e,e')\}$ by examining every line segment in $\Pi(P_i, Q_i)$, and then compute the gradient field of $\textnormal{Area}((P_i + t) \cap Q_i)$ by invoking Lemma~\ref{lem:existsquadratic} on $L$; this takes time $O(b^2)$. Then %determine the constant term in time $O(b^2)$ by
        evaluate $\textnormal{Area}((P_i + \hat{t}) \cap Q_i)$ for an arbitrary point $\hat{t}$ in the interior of $\mathcal{T}'$, using standard algorithms for polygon intersection (see for example \cite{greiner1998efficient}); this also takes time $O(b^2)$. With this information at hand, $g_i(t)$ can be found in $O(1)$ additional time. The total time over all $i \in S_{\textnormal{Bad}}^+$ is $O(\mu b^2)$. %by evaluating the area at an arbitrary point of the relative interior of $\mathcal{T}'$ directly, again splitting nonconvex blocks into $O(1)$ convex sectors if necessary.

        For each $i \in S_{\textnormal{Bad}}^+$, add $g_i(t)$ to $f_{k(i)}^{(1)}(t)$. Observe that for all $k \in [\lceil N/b'\rceil]$, we now have
        \[f_k^{(1)}(t) = \sum_{i \in C_k}\textnormal{Area}((P_i+t) \cap Q_i),\]
        unless there exists at least one index $i \in C_k$ such that $i \in S_{\textnormal{Bad}}^-$. But we only have $i \in S_{\textnormal{Bad}}^-$ when $(P_i + t) \cap Q_i$ does not have the same combinatorial description for all $t$ in the interior of $\mathcal{T}'$. So we can simply return the computed functions $f_k^{(1)}(t)$ for all $k \in [\lceil N/b'\rceil]$ such that $C_k \cap S_{\textnormal{Bad}}^- = \emptyset$, and report the failure condition for the remaining $k \in [\lceil N/b'\rceil]$.
    \end{proof}

    The next step is to compute $\sum_{(i, j) \in C_k \times C_k :  i \neq j} \textnormal{Area}((P_i+t) \cap Q_j)$. The algorithm can be viewed as a counterpart of the one in Lemma \ref{lem:secondsummation}, although the decomposition will be significantly simpler.

    \begin{claim} \label{claim:trans}
        There is an algorithm running in time $O(N \cdot \frac{\log b'}{b})$ that behaves as follows. For each $k \in [\lceil N/b'\rceil]$, the algorithm will either output a %single
        quadratic function $f_k^{(2)}(t)$ such that
        \[f_k^{(2)}(t) = \sum_{\substack{(i, j) \in C_k \times C_k:\\ i \neq j}}\textnormal{Area}((P_i+t) \cap Q_j)\]
        for all $t \in \mathcal{T}'$, or it will correctly determine that there exists an index pair $(i, j) \in C_k \times C_k$ with $i \neq j$ such that $(P_i+t)\cap Q_j$ does not have the same combinatorial description for all $t$ in the interior of $\mathcal{T}'$.
    \end{claim}

    \begin{proof}
        Fix an index $k \in [\lceil N/b'\rceil]$. In the following, we describe how to find $f_k^{(2)}(t)$ (or correctly identify the failure condition) in time $O(\frac{b'\log b'}{b})$. Repeating this procedure for all $O(N/b')$ choices of $k$ gives the desired output and running time.

        Notice that we can decompose the target summation as
        \[\sum_{\substack{(i, j) \in C_k \times C_k :\\ i \neq j}}\!\textnormal{Area}((P_i+t) \cap Q_j)\ = \sum_{i \in C_k}  \textnormal{Area}((P_i+t) \cap Q_{i, k}^-) + \sum_{i \in C_k} \textnormal{Area}((P_i+t) \cap Q_{i, k}^+),\]
        where
        \[Q_{i, k}^- = \bigcup_{j \in C_k : j < i} Q_j \quad \textnormal{ and } \quad Q_{i, k}^+ = \bigcup_{j \in C_k : j > i} Q_j.\]
        Observe that for each $i \in C_k$, $Q_{i, k}^-$ is a union of consecutive blocks of $Q$, and the interval $\mathcal{I}_Q \subseteq [\lceil N/b\rceil]$ containing the indices for these blocks is disjoint from the interval $\mathcal{I}_P$ containing just the block index $i$. (The same is true for each $Q_{i, k}^+$.) Therefore, we can apply Lemma \ref{lem:makequadratic} to evaluate each term in the decomposition in time $O(\log b')$. The total number of terms is $O(b'/b)$, so the total time is $O(\frac{b'\log b'}{b})$. There are two cases:
        \begin{enumerate}
            \item Every invocation of the algorithm in Lemma \ref{lem:makequadratic} succeeds in finding a quadratic function. In this case, we simply add all of the functions together to get a function $f_k^{(2)}(t)$ satisfying the requirements of the claim.
            \item For some $i \in C_k$ and some $\alpha \in \{+, -\}$, the algorithm in Lemma \ref{lem:makequadratic} reports that $(P_i + t) \cap Q_{i, k}^\alpha$ does not have the same combinatorial description for all $t$ in the interior of $\mathcal{T}'$. Because $Q_{i, k}^\alpha$ is a union of consecutive blocks $Q_j$ with $j \in C_k$ and $j \neq i$, this implies that there exists an index pair $(i, j) \in C_k \times C_k$ with $i \neq j$ such that $(P_i + t) \cap Q_j$ does not have the same combinatorial description for all $t$ in the interior of $\mathcal{T}'$. In this case, we can report failure.
        \end{enumerate} 
    % where $Q_{i, k}^-$ is the union of all $Q_j$'s in the original sum with $j < i$, and $Q_{i, k}^+$ is the union of all $Q_j$'s in the original sum with $j > i$. %; if either union is empty, its area contribution is zero.
    % Now for each $k \in [\lceil N/b'\rceil]\backslash S_{\textnormal{Bad}}'$, invoke Lemma \ref{lem:makequadratic} for each of the terms individually. If the lemma successfully finds a quadratic function for every term, % and no diagonal term has failed, 
    % then we add $k$ to $S_{\textnormal{Good}}'$ and store $f_k'$ as the sum of all diagonal and off-diagonal quadratic functions for this $k$. Otherwise, the lemma fails for at least one of the terms, or a diagonal term has already failed, and we add the corresponding index $k$ to $S_{\textnormal{Bad}}'$. The total time for each $k$ is $O(\frac{b'}{b}\log b')$, so the total time over all $O(N/b')$ choices of $k$ is $O(N \frac{\log b'}{b})$.
    \end{proof}

    With Claims \ref{claim:cis} and \ref{claim:trans} at hand, we are nearly done. Initialize $S_{\textnormal{Good}}' = \emptyset$ and $S_{\textnormal{Bad}}' = \emptyset$. Execute the algorithms in both claims, which takes time $O(N \cdot \frac{\log b'}{b} + \mu b^2)$. For every $k \in [\lceil N/b'\rceil]$ such that the algorithm in Claim \ref{claim:cis} output a function $f_k^{(1)}(t)$ and the algorithm in Claim \ref{claim:trans} output a function $f_k^{(2)}(t)$, observe that
    \begin{align*}
        \textnormal{Area}\left((P_k'+t) \cap Q_k'\right) & = \sum_{i \in C_k}\textnormal{Area}((P_i+t) \cap Q_i)\ + \!\sum_{\substack{(i, j) \in C_k \times C_k:\\ i \neq j}}\!\!\!\!\textnormal{Area}((P_i+t) \cap Q_j). \\
        & = f_k^{(1)}(t) + f_k^{(2)}(t)
    \end{align*}
    for all $t \in \mathcal{T}'$. Add every such index $k$ to $S_{\textnormal{Good}}'$, and store the function $f_k'(t) = f_k^{(1)}(t) + f_k^{(2)}(t)$. Add all of the remaining indices $k$ to $S_{\textnormal{Bad}}'$. The additional time is $O(N/b')$.

    All that remains is to reason about $\vert S_{\textnormal{Bad}}'\vert$. If we have $\vert S_{\textnormal{Bad}}'\vert \leq \mu'$, then the data structure we constructed satisfies all of the requirements of a $(\mu', b', \mathcal{T}')$-block structure, and we are done. Now assume that $\vert S_{\textnormal{Bad}}'\vert > \mu'$. We need to argue that $\mathcal{T}'$ is strictly intersected by more than $\mu'$ lines from the multiset
    \begin{equation*}
        \mathcal{S}\ := \biguplus_{k \in [\lceil N/b'\rceil]} \left(\biguplus_{(i, j) \in C_k \times C_k}\overleftrightarrow{\Pi}(P_{i}, Q_{j})\right),
    \end{equation*}
    in which case we can simply return failure.

    Consider a single index $k \in S_{\textnormal{Bad}}'$. By construction, there is at least one pair $(i, j) \in C_k \times C_k$ such that $(P_i+t)\cap Q_j$ does not have the same combinatorial description for all $t$ in the interior of $\mathcal{T}'$. Thus for some edge pair $(e,e')$ from $(P_i,Q_j)$, neither of the following holds:
    \begin{enumerate}
        \item $\mathcal{T}' \subseteq \pi(e, e')$.
        \item $\pi(e, e')$ is disjoint from the interior of $\mathcal{T}'$.
    \end{enumerate}
    This implies that there exists a line $\mathcal{L}$ in $\overleftrightarrow{\Pi}(P_i, Q_j)$ that strictly intersects $\mathcal{T}'$. Repeating the argument for all $k \in S_{\textnormal{Bad}}'$, we have a set of at least $\vert S_{\textnormal{Bad}}'\vert > \mu'$ lines from the multiset $\mathcal{S}$ that strictly intersect $\mathcal{T}'$.
\end{proof}

\section{The Linear Time Algorithm} \label{sec:linalg}

So far, we showed how to use block structures to efficiently solve \textsc{MaxRegion} when the input $(\mu, b, \mathcal{T})$-block structure has $\mathcal{T}$ being just a single point, and we showed how to efficiently construct new block structures from existing ones. In this section, we describe how to efficiently solve \textsc{MaxRegion} when the input $(\mu, b, \mathcal{T})$-block structure has $\mathcal{T}$ being a triangle or line segment, and then combine all of the subroutines to get our linear time algorithm for Problem \ref{prob:overlap}. Recall the statement of \textsc{MaxRegion}: %The algorithm works by invoking \textsc{MaxRegion} on regions with lower dimension to iteratively shrink $\mathcal{T}$. %\textcolor{red}{Make it clear that the sequence of triangles/line segments that we will construct is different from the one in Section 3.2. Also just re-write the proof to ``sound better,'' although it looks like no major changes are needed.}

\MaxRegionProblem*

Recall that $T_{d\textsc{-Max}}(N, b, \mu)$ denotes the expected running time of the fastest algorithm for \textsc{MaxRegion}, where $d$ is the dimension of $\mathcal{T}$.

We start by showing how to reduce from \textsc{MaxRegion} over a triangle to (i) a few instances of \textsc{MaxRegion} over a line segment, and (ii) a single instance of \textsc{MaxRegion} over a triangle, but with a more refined block structure.

\begin{lemma} \label{lem:subroutineT2}
    For all $2 \leq b, b' \leq N^{o(1)}$ such that $b$ divides $b'$, for all $\mu, \mu'$ such that $\mu' = N^{1-o(1)}$, %we have:
    \begin{align*}
        T_{2\textsc{-Max}}(N, b, \mu)\ \leq\ \ & O(\log (Nb'/\mu')) \cdot T_{1\textsc{-Max}}(N, b, \mu) + O(N/b') + {} \\
        & O(1) \cdot T_{\textsc{Cascade}}(N, b, b', \mu) + T_{2\textsc{-Max}}(N, b', \mu')\\[4pt]
    %\end{align*}
    %\begin{align*}
    \end{align*}
\end{lemma}

% T_{1\textsc{-Max}}(N, b, \mu)\ \leq\ \ & O(\log (Nb'/\mu')) \cdot T_{0\textsc{-Max}}(N, b, \mu) + O(N/b) + {}\\
%         & O(1) \cdot T_{\textsc{Cascade}}(N, b, b', \mu) + T_{1\textsc{-Max}}(N, b', \mu').

\begin{proof}
    Let $P_1, \ldots, P_{\lceil N/b\rceil}$ and $Q_1, \ldots, Q_{\lceil N/b\rceil}$ be the blocks of $P$ and $Q$ determined by parameter~$b$, and let $P'_1, \ldots, P'_{\lceil N/b'\rceil}$ and $Q'_1, \ldots, Q'_{\lceil N/b'\rceil}$ be the blocks determined by parameter~$b'$. As in Section \ref{sec:cascade}, define $C_k := [(k-1)\cdot \frac{b'}{b} + 1, k \cdot \frac{b'}{b}]$ for $1 \leq k < \lceil N/b'\rceil$ and $C_{\lceil N/b'\rceil} := [(\lceil N/b'\rceil-1)\cdot \frac{b'}{b} + 1, \lceil N/b\rceil]$. The bulk of the algorithm lies in finding a well-behaved triangle $\mathcal{T}' \subseteq \mathcal{T}$:

\begin{claim} \label{claim:dimensionrecurse}
    For all $2 \leq b, b' \leq N^{o(1)}$ such that $b$ divides $b'$, for all $\mu, \mu'$ such that $\mu' = N^{1-o(1)}$, there is an algorithm running in expected time $O(\log (Nb'/\mu')) \cdot T_{1\textsc{-Max}}(N, b, \mu) + O(N/b')$ that outputs a triangle $\mathcal{T}' \subseteq \mathcal{T}$ such that the following holds.
    \begin{enumerate}
        \item The point $t^* \in \mathcal{T}$ maximizing $\textnormal{Area}((P+t^*) \cap Q)$ satisfies $t^* \in \mathcal{T}'$.
        \item With high probability, $\mathcal{T}'$ is strictly intersected by at most $\mu'$ lines from the multiset
    \begin{equation*}
        \mathcal{S}\ := \biguplus_{k \in [\lceil N/b'\rceil]} \left(\biguplus_{(i, j) \in C_k \times C_k}\overleftrightarrow{\Pi}(P_{i}, Q_{j})\right).
    \end{equation*}
    \end{enumerate}
\end{claim}

\begin{proof}
    Compute the set $S'_{\textnormal{nonempty}} \subseteq [\lceil N/b'\rceil]$ of all indices $k$ such that both $P_k'$ and $Q_k'$ are nonempty, which takes time $O(N/b')$. Let $r$ be any constant such that
    \[\left\vert \biguplus_{(i, j) \in C_k \times C_k}\overleftrightarrow{\Pi}(P_{i}, Q_{j})\right\vert \leq r(b')^2\]
    for all $k \in [\lceil N/b'\rceil]$. Define $\mathcal{T}_0 := \mathcal{T}$.
    
    Our plan is to find a series of triangles $\mathcal{T}_0 \supseteq \mathcal{T}_1 \supseteq \mathcal{T}_2 \supseteq \ldots$ such that (i) each triangle is guaranteed to contain the optimal translation $t^*$ for the original region $\mathcal{T}$, and (ii) with high probability, the number of lines from $\mathcal{S}$ that strictly intersect each triangle decreases geometrically. The number of lines strictly intersecting $\mathcal{T}_0$ is at most $\vert \mathcal{S}\vert \leq O(Nb')$, so after $O(\log (Nb'/\mu'))$ iterations, we will have a triangle that is (with high probability) strictly intersected by at most $\mu'$ lines of $\mathcal{S}$. This triangle can be taken as $\mathcal{T}'$. So to prove the claim, it will suffice to argue that each triangle in the sequence can be found in expected time $O(1) \cdot T_{1\textsc{-Max}}(N, b, \mu) + O(N/(b'\log (Nb'/\mu')))$.

    \begin{figure}
    \centering
    \includegraphics[width=0.5\textwidth]{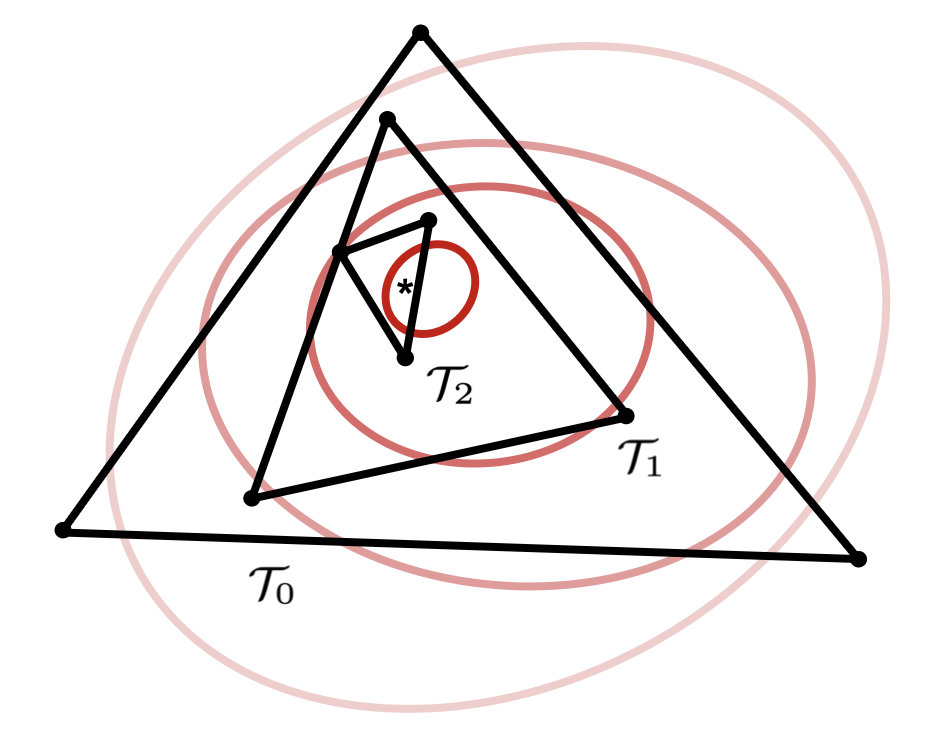}
    \caption{A series of triangles $\mathcal{T}_0$, $\mathcal{T}_1$, and $\mathcal{T}_2$ as produced by the algorithm. The (unknown) optimal placement $t^*$ is denoted by an asterisk. The red curves indicate contour lines of the objective function, with darker circles indicating larger values of $\textnormal{Area}((P+t) \cap Q)$.}
    \label{fig:Tseries}
\end{figure}
    
    Consider producing a triangle $\mathcal{T}_{\ell+1}$ from a triangle $\mathcal{T}_\ell$. We start by using Lemma \ref{cor:approxcut} to make, with high probability, a $\frac{1}{2}$-cutting for the lines of $\mathcal{S}$ that strictly intersect $\mathcal{T}_\ell$.
    
    To apply this lemma, we need an efficient sampler. We argue that rejection sampling will suffice. First, sample a uniformly random index $k \in S'_{\textnormal{nonempty}}$, and explicitly construct
    \[\mathcal{S}_k = \biguplus_{(i, j) \in C_k \times C_k}\overleftrightarrow{\Pi}(P_{i}, Q_{j}).\]
    Then sample a uniformly random line from $\mathcal{S}_k$. If the proposed line does not strictly intersect $\mathcal{T}_\ell$, then reject it. Otherwise, accept the line with probability $\vert \mathcal{S}_k\vert/(r(b')^2)$. Observe that the expected time required to sample a line is $(b')^{O(1)} / \delta = N^{o(1)} / \delta$, where $\delta$ is the fraction of lines in $\mathcal{S}$ that strictly intersect $\mathcal{T}_\ell$.

    Suppose that rejection sampling fails to find a random line of $\mathcal{S}$ that strictly intersects $\mathcal{T}_\ell$ within time $O(N^{0.1})$. This implies that $\delta \leq N^{-0.1+o(1)}$ with high probability. Because $b' \leq N^{o(1)}$, we know that $\vert \mathcal{S}\vert = N^{1+o(1)}$. This implies that, with high probability, only $N^{0.9 + o(1)}$ lines of $\mathcal{S}$ strictly intersect $\mathcal{T}_\ell$. Since $\mu' = N^{1 - o(1)}$, we can directly take $\mathcal{T}' = \mathcal{T}_\ell$.
    
    If rejection sampling is fast enough, we can produce the required $\frac{1}{2}$-cutting with high probability in time $O(\vert \mathcal{S}\vert^{0.1+o(1)}) = O(N^{0.1+o(1)}) \leq O(N/(b'\log (Nb'/\mu')))$. Then, we intersect this cutting with $\mathcal{T}_\ell$ and re-triangulate the resulting cells, which takes $O(1)$ time in total.
    
    The final step is to locate the cell of this triangulation that contains the optimal translation $t^*$. This cell will be $\mathcal{T}_{\ell+1}$; by definition of the $\frac{1}{2}$-cutting, the interior of $\mathcal{T}_{\ell+1}$ will intersect at most half as many lines of $\mathcal{S}$ as $\mathcal{T}_\ell$. All we need to do is check all $O(1)$ triangles by recursing in the dimension. Consider a single triangle $\mathcal{X}$. Let $\mathcal{X}^-$ be an infinitesimally smaller copy of $\mathcal{X}$ that is strictly inscribed within $\mathcal{X}$. %\footnote{Infinitesimal shifts are a standard trick, and all of the algorithms in this paper can handle them without modification as long as they are implemented carefully enough.}
    All $O(1)$ bounding line segments for $\mathcal{X}$ and $\mathcal{X}^-$ are contained within $\mathcal{T}$. Therefore, we can use the input $(\mu, b, \mathcal{T})$-block structure as a block structure for these line segments, and find the optimal translation (along with the area of intersection) restricted to each line segment in time $O(1) \cdot T_{1\textsc{-Max}}(N, b, \mu)$. Let $\textnormal{OPT}(\partial\mathcal{X})$ be $\textnormal{Area}((P+\tilde{t})\cap Q)$ for the optimal translation $\tilde{t}$ on the boundary of $\mathcal{X}$, and similarly for $\textnormal{OPT}(\partial\mathcal{X}^-)$. There are four cases:
    \begin{enumerate}
        \item There is a triangle $\mathcal{X}$ such that $\textnormal{OPT}(\partial\mathcal{X}^-) > \textnormal{OPT}(\partial\mathcal{X})$. By the concavity of $\sqrt{\textnormal{Area}((P+t)\cap Q)}$ (Lemma \ref{lem:objconcave}), this triangle must contain $t^*$.
        \item There is a triangle $\mathcal{X}$ such that $\textnormal{OPT}(\partial\mathcal{X}^-) = \textnormal{OPT}(\partial\mathcal{X}) > 0$. Again using Lemma \ref{lem:objconcave}, the optimal translation $\tilde{t}$ on the boundary of $\mathcal{X}$ must satisfy $\textnormal{Area}((P+\tilde{t})\cap Q) = \textnormal{Area}((P+t^*)\cap Q)$. In this case, we stop the recursion and return the degenerate triangle $\mathcal{T}' = \tilde{t}$.
        \item There is a triangle $\mathcal{X}$ such that $\textnormal{OPT}(\partial\mathcal{X}^-) < \textnormal{OPT}(\partial\mathcal{X})$, and every other triangle $\mathcal{X}'$ has either $\textnormal{OPT}(\partial\mathcal{X}'^-) < \textnormal{OPT}(\partial\mathcal{X}')$ or $\textnormal{OPT}(\partial\mathcal{X}'^-) = \textnormal{OPT}(\partial\mathcal{X}') = 0$. By Lemma \ref{lem:objconcave}, the triangle $\mathcal{X}''$ maximizing $\textnormal{OPT}(\partial\mathcal{X}'')$ must have that the optimal translation $\tilde{t}''$ on its boundary satisfies $\textnormal{Area}((P+\tilde{t}'')\cap Q) = \textnormal{Area}((P+t^*)\cap Q)$. We again stop the recursion and return $\mathcal{T}' = \tilde{t}''$.
        \item For all triangles $\mathcal{X}$, we have $\textnormal{OPT}(\partial\mathcal{X}'^-) = \textnormal{OPT}(\partial\mathcal{X}') = 0$. In this case, we find an arbitrary translation $t \in \mathbb{R}^2$ such that $(P+t)\cap Q \neq \emptyset$, and set $\mathcal{T}_{\ell+1}$ to be the triangle $\mathcal{X}$ such that $t \in \mathcal{X}$. (If no such triangle exists, then $\textnormal{Area}((P+t^*)\cap Q) = 0$. We stop the recursion and return $\mathcal{T}' = \tilde{t}$ for an arbitrary $\tilde{t} \in \mathcal{T}$.)
    \end{enumerate}
\end{proof}

We now bound $T_{2\textsc{-Max}}(N, b, \mu)$. Invoke the algorithm in Claim \ref{claim:dimensionrecurse}, and then invoke \textsc{Cascade} using the resulting triangle $\mathcal{T}'$. If \textsc{Cascade} successfully creates a $(\mu', b', \mathcal{T}')$-block structure, then we recursively find the optimal translation in $\mathcal{T}'$ in time $T_{2\textsc{-Max}}(N, b', \mu')$. If \textsc{Cascade} returns failure, then we simply invoke the algorithm in Claim \ref{claim:dimensionrecurse} again to get a new triangle $\mathcal{T}'$ and try again. Because \textsc{Cascade} can only return failure when $\mathcal{T}'$ is strictly intersected by more than $\mu'$ lines from $\mathcal{S}$, the expected number of attempts is $O(1)$ (indeed, $1 + N^{-\omega(1)}$).
\end{proof}

A similar lemma holds for \textsc{MaxRegion} over line segments:

\begin{lemma} \label{lem:subroutineT1}
    For all $2 \leq b, b' \leq N^{o(1)}$ such that $b$ divides $b'$, for all $\mu, \mu'$ such that $\mu' = N^{1-o(1)}$, %we have:
    \begin{align*}
        T_{1\textsc{-Max}}(N, b, \mu)\ \leq\ \ & O(\log (Nb'/\mu')) \cdot T_{0\textsc{-Max}}(N, b, \mu) + O(N/b') + {}\\
        & O(1) \cdot T_{\textsc{Cascade}}(N, b, b', \mu) + T_{1\textsc{-Max}}(N, b', \mu').
    \end{align*}
\end{lemma}

\begin{proof}
    Nearly identical to the proof of Lemma \ref{lem:subroutineT2}. We produce a series of line segments $\mathcal{T}_0 \supseteq \mathcal{T}_1 \supseteq \mathcal{T}_2 \supseteq \ldots$, and construct cuttings along the line parallel to $\mathcal{T}_0$ as opposed to the entirety of $\mathbb{R}^2$.
\end{proof}

\paragraph{Putting It All Together.}
We now present our linear time algorithm. This amounts to solving a curious recurrence, and we actually have significant freedom in choosing the values for $b$ and $b'$.

\begin{theorem}
    There is an algorithm for Problem \ref{prob:overlap} running in expected time $O(n+m)$.
\end{theorem}

\begin{proof}
We will always take $\mu := 20 \cdot N/b^3$ and $\mu' := 20 \cdot N/(b')^3$, so we can drop the $\mu$ and $\mu'$ arguments from the time complexity of each subroutine/oracle. We begin by reducing from Problem~\ref{prob:overlap} to an instance of \textsc{MaxRegion} over a $(20 \cdot N/2^3, 2, \mathcal{T})$-block structure\footnote{Notice that taking $\mu = 20 \cdot N/2^3$ actually exceeds the total number of block indices. This is well-defined because $\mu$ is just an upper bound on $\vert S_{\textnormal{Bad}}\vert$.}
where $\mathcal{T}$ is a triangle. The partition of $S$ will simply be $S_{\textnormal{Bad}} = [\lceil N/2\rceil]$ and $S_{\textnormal{Good}} = \emptyset$. So the only work required for initialization is to: (i) store the vertices/edges of $P$ and $Q$ appropriately, (ii) compute the area prefix sums for $P$ and $Q$ once and for all, (iii) sort the multiset
\[U = \{\theta_P(e) : e \textnormal{ is an edge of $P$}\} \,\cup\, \{\theta_Q(e) : e \textnormal{ is an edge of $Q$}\}\]
once and for all, and (iv) identify any triangle $\mathcal{T}$ such that, for all $t \in \mathbb{R}^2$, if $(P+t)\cap Q$ is nonempty then $t \in \mathcal{T}$. All steps can be performed in linear time because the vertices of $P$ and $Q$ are already given in sorted order. The total expected running time for Problem \ref{prob:overlap} is thus $O(N) + T_{2\textsc{-Max}}(N, 2)$.

Before analyzing $T_{2\textsc{-Max}}(N, 2)$, we find a closed-form expression for $T_{1\textsc{-Max}}(N, b)$ when $2 \leq b \leq N^{o(1)}$. As per Lemma \ref{lem:subroutineT1}, we have the following for any $2 \leq b, b' \leq N^{o(1)}$ such that $b$ divides $b'$:
\begin{align*}T_{1\textsc{-Max}}(N, b)\ \leq\ \ & O(\log b') \cdot T_{0\textsc{-Max}}(N, b) + O(N/b') + {}\\
        & O(1) \cdot T_{\textsc{Cascade}}(N, b, b') + T_{1\textsc{-Max}}(N, b').\end{align*}
We know by Lemma \ref{lem:Cascade} that
\[T_{\textsc{Cascade}}(N, b, b') \leq O\left(N \cdot \frac{\log b'}{b}\right),\]
and we know by Lemma \ref{lem:T0} that
\[T_{0\textsc{-Max}}(N, b) \leq O\left(N \cdot \frac{\log b}{b}\right).\]
Putting everything together, we have the following for any $2 \leq b, b' \leq N^{o(1)}$ such that $b$ divides $b'$:
\[T_{1\textsc{-Max}}(N, b) \leq O\left(N \cdot \frac{\log b'\log b}{b}\right) + T_{1\textsc{-Max}}(N, b').\]

We can efficiently parallelize the algorithm for zero-dimensional \textsc{MaxRegion} given in the proof of Lemma \ref{lem:T0}. So via a standard application of multidimensional parametric search (see Theorem 4.4 in the full version of \cite{Toledo92}) that uses zero-dimensional \textsc{MaxRegion} for the decision problem, $T_{1\textsc{-Max}}(N, b) \leq O(N/\log N)$
when $b \geq (\log N)^{c_1}$ for some constant $c_1$. If we take $b' = 2b$ (for example) and solve the recurrence, we get
%\todo{TC: I've changed $b'=b^2$ to $b'=2b$ (here and below) to be consistent with your sum}
\[T_{1\textsc{-Max}}(N, b) \leq O\left(\frac{N}{\log N} \ + \sum_{a = \log b}^{O(\log \log N)} N \cdot\frac{\log^2 (2^a)}{2^a}\right) = O\left(N \cdot \frac{\log^2 b}{b} + \frac{N}{\log N}\right).\]

Now we turn our attention to $T_{2\textsc{-Max}}(N, b)$. Using Lemma \ref{lem:subroutineT2}, we have the following for any $2 \leq b, b' \leq N^{o(1)}$ such that $b$ divides $b'$:
\begin{align*}
        T_{2\textsc{-Max}}(N, b)\ \leq\ \ & O(\log b') \cdot T_{1\textsc{-Max}}(N, b) + O(N/b') + {} \\
        & O(1) \cdot T_{\textsc{Cascade}}(N, b, b') + T_{2\textsc{-Max}}(N, b')
\end{align*}
Combining with Lemma \ref{lem:Cascade} and our closed form expression for $T_{1\textsc{-Max}}(N, b)$, we have (assuming the same restrictions on $b$ and $b'$):
\[T_{2\textsc{-Max}}(N, b) \leq O\left(N \cdot \frac{\log b'\log^2 b}{b} + N \cdot \frac{\log b'}{\log N}\right) + T_{2\textsc{-Max}}(N, b').\]

Using another application of multidimensional parametric search that uses zero-dimensional $\textsc{MaxRegion}$ for the decision problem, we have $T_{2\textsc{-Max}}(N, b) \leq O(N/\log N)$
when $b \geq (\log N)^{c_2}$ for some constant $c_2$. Again taking $b' = 2b$, we thus have
\begin{align*}
    T_{2\textsc{-Max}}(N, b) \leq & O\left(\frac{N}{\log N} \ + \sum_{a = \log b}^{O(\log \log N)} \left( N \cdot\frac{\log^3 (2^a)}{2^a} + N \cdot \frac{\log (2^a)}{\log N}\right)\right) \\
    \:=\: & O\left(N \cdot \frac{\log^3 b}{b} + N \cdot \frac{\log^2\log N}{\log N}\right).
\end{align*}
% If we again take $b' = 2b$ and stop to apply multidimensional parametric search when $b' = (\log N)^{c_2}$ for a sufficiently large constant $c_2$, we have $O(\log\log N)$ recursion depth, and $\log b'$ never exceeds $O(\log\log N)$. 
% Solving the recurrence in the same way as for $T_{1\textsc{-Max}}(N, b)$ gives a loose upper bound of:
% \[T_{2\textsc{-Max}}(N, b) \leq O\left(N \cdot \frac{\log^4 b}{b} + N \cdot \frac{\log^2\log N}{\log N}\right).\]
Plugging in $b = 2$, we have $T_{2\textsc{-Max}}(N, 2) \leq O(N) = O(n+m)$, and the theorem follows.
\end{proof}

\small
\bibliographystyle{plainurl}
\bibliography{refs}

\end{document}